%% file: submission.tex
\documentclass[letterpaper,USenglish]{lipics-v2016}

\usepackage{microtype} 

\usepackage{amssymb,amsmath}
\usepackage{hyperref}
\hypersetup{
  unicode=true,
  pdfkeywords={session types; linear types; dependent types; lambda calculus},
  pdftitle={Dependent Session Types},
  pdfauthor={H. Wu and H. Xi},
  colorlinks=true,
  linkcolor=red,
  citecolor=green,
  urlcolor=cyan,
  breaklinks=true
}

\usepackage{ifxetex,ifluatex}
\usepackage{fixltx2e} 
\ifnum 0\ifxetex 1\fi\ifluatex 1\fi=0 
  \usepackage[T1]{fontenc}
  \usepackage[utf8]{inputenc}
\else 
  \ifxetex
    \usepackage{mathspec}
  \else
    \usepackage{fontspec}
  \fi
  \defaultfontfeatures{Ligatures=TeX,Scale=MatchLowercase}
\fi
\IfFileExists{upquote.sty}{\usepackage{upquote}}{}

\usepackage{color}
\usepackage{fancyvrb}

\DefineVerbatimEnvironment{Highlighting}{Verbatim}{commandchars=\\\{\}}
\newenvironment{Shaded}{}{}
\newcommand{\KeywordTok}[1]{\textcolor[rgb]{0.00,0.44,0.13}{\textbf{#1}}}
\newcommand{\DataTypeTok}[1]{\textcolor[rgb]{0.56,0.13,0.00}{#1}}
\newcommand{\DecValTok}[1]{\textcolor[rgb]{0.25,0.63,0.44}{#1}}

\newcommand{\CommentTok}[1]{\textcolor[rgb]{0.38,0.63,0.69}{\textit{#1}}}

\newcommand{\NormalTok}[1]{#1}

\usepackage{colonequals} 
\usepackage[a]{esvect} 
\usepackage{prftree}
\usepackage{adjustbox} 

\usepackage[capitalise,noabbrev,nameinlink]{cleveref} 
\crefname{lemma}{Lemma}{Lemmas}
\crefname{ex}{Example}{Examples}
\crefname{app}{Appendix}{Appendix}

\usepackage{pgfmath}
\usepackage{tikz}
\usetikzlibrary{arrows,shapes,petri}

\providecommand{\tightlist}{%
  \setlength{\itemsep}{0pt}\setlength{\parskip}{0pt}}

\title{Dependent Session Types}
\titlerunning{Dependent Session Types}

\author[1]{Hanwen Wu}
\affil[1]{Boston University\\\texttt{hwwu@cs.bu.edu}}
\author[2]{Hongwei Xi}
\affil[2]{Boston University\\\texttt{hwxi@cs.bu.edu}}
\authorrunning{H. Wu and H. Xi}

\Copyright{Hanwen Wu and Hongwei Xi}



\EventEditors{John Q. Open and Joan R. Acces}
\EventNoEds{2}
\EventLongTitle{28th International Conference on Concurrency Theory}
\EventShortTitle{CONCUR 2017}
\EventAcronym{CONCUR}
\EventYear{2017}
\EventDate{September 4--9, 2017}
\EventLocation{Berlin, German}
\EventLogo{}
\SeriesVolume{42}
\ArticleNo{23}

\begin{document}


\newcommand{\Lzero}[0]{\mathcal{L}_0}
\newcommand{\Ldep}[0]{\mathcal{L}_{\forall,\exists}}
\newcommand{\Lchan}[0]{\mathcal{L}_{\forall,\exists}^{\pi}}

\def\ATS{$\mathcal{A}\mathcal{T}\kern-.1em\mathcal{S}$}
\def\res{\mathcal{R}}

\newcommand{\defeq}[0]{\coloncolonequals}
\newcommand{\dom}[0]{\text{\rm\bf dom}}

\newcommand{\sort}[1]{\text{\rm\it #1}}
\newcommand{\type}[1]{\text{\rm\bf #1}}
\newcommand{\stype}[1]{\text{\rm\tt #1}}
\newcommand{\subtype}[0]{\leq_{ty}}
\newcommand{\sta}[1]{\text{\rm\it #1}}
\newcommand{\dyn}[1]{\text{\rm\bf #1}}
\newcommand{\guard}[0]{\,{\supset}\,}
\newcommand{\guardi}[0]{{\supset^+}}
\newcommand{\guarde}[0]{{\supset^-}}
\newcommand{\foralli}[0]{\forall^+}
\newcommand{\foralle}[0]{\forall^-}
\newcommand{\assert}[0]{{\wedge}}

\newcommand{\sig}[0]{\mathcal{S}}
\newcommand{\dcx}[0]{\textit{dcx}}
\newcommand{\dcc}[0]{\textit{dcc}}
\newcommand{\dcf}[0]{\textit{dcf}}
\newcommand{\dcr}[0]{\textit{dcr}}
\newcommand{\scx}[0]{\textit{scx}}
\newcommand{\scc}[0]{\textit{scc}}
\newcommand{\scf}[0]{\textit{scf}}

\newcommand{\chan}[0]{\type{chan}}
\newcommand{\msg}[0]{\stype{msg}}
\newcommand{\quan}[0]{\stype{quan}}
\newcommand{\fix}[0]{\stype{fix}}
\newcommand{\nil}[0]{\stype{end}}
\newcommand{\branch}[0]{\stype{branch}}
\newcommand{\ite}[0]{\stype{ite}}


\newcommand{\cli}[0]{\stype{C}}
\newcommand{\srv}[0]{\stype{S}}


\makeatletter
\setlength{\@fptop}{0pt}
\setlength{\@fpsep}{\floatsep}
\setlength{\@fpbot}{0pt plus 1fil}
\makeatother

\maketitle

\begin{abstract}
Session types offer a type-based discipline for enforcing communication
protocols in distributed programming. We have previously formalized
simple session types in the setting of multi-threaded
\(\lambda\)-calculus with linear types. In this work, we build upon our
earlier work by presenting a form of dependent session types (of
DML-style). The type system we formulate provides linearity and duality
guarantees with no need for any runtime checks or special encodings. Our
formulation of dependent session types is the first of its kind, and it
is particularly suitable for practical implementation. As an example, we
describe one implementation written in ATS that compiles to an
Erlang/Elixir back-end.
\end{abstract}

\section{Introduction}\label{sec:introduction}

A session is a sequence of interactions among concurrently running
programs. We assign session types
\cite{Honda:1993eh,Honda:1998fm,Takeuchi:1994bv,Honda:2008hi} to
communication channels to ensure session fidelity, which means each
participant in the session communicates according to a chosen protocol.
Recent works
\cite{Wadler:2012ua,Wadler:2014bta,Caires:2010gi,Carbone:2015hl,Carbone:2016kd}
have established a form of Curry-Howard correspondence where logical
propositions are interpreted as session types for terms in variants of
\(\pi\)-calculus \cite{Milner:1992fc,Milner:1992gv}.

Instead of \(\pi\)-calculus, it is also possible to formulate session
types in the setting of \(\lambda\)-calculus
\cite{Lindley:2015cr,Xi:2016ue}. This paper formulates a form of
dependent session types by extending our prior work \cite{Xi:2016ue}.

More specifically, the formulation is based on Applied Type Systems
(\ATS~\cite{Xi:2003kl,Xi:2004vt}), a type system supporting dependent
types (of DML-style \cite{Xi:1999bh}), linear types, and programming
with theorem proving. \ATS~takes a layered approach to dependent types
in which \emph{statics}, where types are formed and reasoned about, are
completely separate from \emph{dynamics}, where programs are constructed
and evaluated. Based on \ATS, session protocols are then captured by
extending statics with session types (static terms of sort
\(\sort{stype}\)), while communication channels are \emph{linear}
dynamic values whose types are \emph{indexed} by such session protocols.

When compared to other similar works (e.g. \cite{Toninho:2011in}), a
very important difference of our formulation is that our session types
describe the intended behavior \emph{globally}, instead of using a
polarized presentation where dual session types are used to describe
dual endpoints of a channel \emph{locally}. This is especially so when
quantifiers are involved.

Suppose that we want to provide an equality testing service, which
receives two integers \(m\) and \(n\), and then sends out a boolean
value indicating whether they are equal. Let us use \emph{roles} 0
(server) and 1 (client), to refer to the two endpoints of a channel. We
may use \texttt{S} for 0 and \texttt{C} for 1. We use \texttt{equal} for
the following (static) term which describes the protocol for the
equality testing service,
\[\stype{equal}\defeq\msg(\cli,\type{int})\coloncolon\msg(\cli,\type{int})\coloncolon\msg(\srv,\type{bool})\coloncolon\nil(\srv)\]

We use \(\msg(r,\hat\tau)\) to mean that the endpoint \(r\) (more
precisely, the party holding endpoint \(r\)) is to send a (linear) value
of type \(\hat\tau\), and \(\coloncolon\) for chaining, and \(\nil(r)\)
to mean that the endpoint \(r\) is to initiate the termination of the
session (while the other side waits for it). With dependent session
types, \texttt{equal} can be given a more precise definition as follows,
\begin{multline*}
\stype{equal}\defeq\quan(\cli,\lambda m{:}\sort{int}.\quan(\cli,\lambda n{:}\sort{int}.\\
\msg(\cli,\type{int}(m))\coloncolon\msg(\cli,\type{int}(n))\coloncolon\msg(\srv,\type{bool}(m=n))\coloncolon\nil(\srv)))
\end{multline*}
where \(\quan\) is a \emph{global} encoding of quantifiers. For any role
\(r\), \(\quan(r,\cdot)\) means universal quantification at endpoint
\(r\), and dually, existential quantification at the other endpoint
(\(1-r\)). In \texttt{equal}, \texttt{quan} means universal
quantification at the client side, meaning the client process can send
any integers onto the endpoint. Dually, \texttt{quan} refers to
existential quantification at the server side, indicating that the
server process can only send back a boolean value representing the
equality of the two \emph{received} integers. Note that \(\type{int}\)
and \(\type{bool}\) are type constructors (static functions of c-sort
\(\sort{int}\Rightarrow\sort{type}\) and
\(\sort{bool}\Rightarrow\sort{type}\), respectively) while
\(\sort{int}\) and \(\sort{bool}\) are \emph{sorts} for static terms.
Both \(\type{int}(i)\) and \(\type{bool}(b)\) are singleton types
representing values that equal \(i\) and \(b\), respectively. In ATS,
which uses ML-like syntax, an example program of the type
\(\chan(\srv,\stype{equal})\rightarrow\type{1}\) that provides such
service on the server side endpoint can be written as follows,

\begin{Shaded}
\begin{Highlighting}[]
\KeywordTok{fun}\NormalTok{ eq_test (ch:chan(S,equal)): void = }\KeywordTok{let} 
    \KeywordTok{prval}\NormalTok{ () = exify ch }\CommentTok{(* prval denotes a proof value, that will be *)}
    \KeywordTok{prval}\NormalTok{ () = exify ch }\CommentTok{(* erased after type-checking                *)}
      \KeywordTok{val}\NormalTok{  m = recv  ch }
      \KeywordTok{val}\NormalTok{  n = recv  ch }
      \KeywordTok{val}\NormalTok{ () = send (ch, m = n)}
\KeywordTok{in}\NormalTok{ close ch }\KeywordTok{end}
\end{Highlighting}
\end{Shaded}

Let us use this code sample to introduce some key concepts. We use
(linear) \emph{channels} for communication. A channel consists of two
\emph{endpoints}. When one process sends a value onto one endpoint, the
value is automatically transmitted to the other endpoint of the channel.
\texttt{ch} is one such endpoint of the channel at party \texttt{S},
whose type is \(\chan(\srv,\stype{equal})\). The \emph{linear} type
constructor, \(\chan\), will construct a linear type \(\chan(r,\pi)\)
given a role \(r\) and a global session type \(\pi\). The combination of
\(r\) and \(\pi\) is where a global session type gets ``projected''
locally. This can be used to type an endpoint of a channel at party
\(r\). As \texttt{equal} is globally quantified by session type
constructor \texttt{quan}, we need to locally interpret it at party
\texttt{S}, by calling a session API \texttt{exify} twice, which
essentially turns \(\chan(\srv,\stype{equal})\) into \[
\exists m{:}\sort{int}.\exists n{:}\sort{int}.
\chan(\srv,\msg(\cli,\type{int}(m))\coloncolon\msg(\cli,\type{int}(n))\coloncolon\msg(\srv,\type{bool}(m=n))\coloncolon\nil(\srv))
\] for use with other session API, e.g. \texttt{recv}. The \emph{guard}
in the signature of \texttt{exify} (see \cref{fig:sessionapi}),
\(r\neq r_0\), specifies that, for any \(\quan(r_0,\cdot)\) at endpoint
\(\chan(r,\cdot)\), only when \(r\neq r_0\) is true that \texttt{exify}
can be invoked to turn \(\chan(r,\quan(r_0,\cdot))\) into
\(\exists a{:}\sigma.\chan(r,\cdot)\). Dually, before the client can use
the channel to send two integers, it has to locally interpret
\texttt{quan} at party \texttt{C}, by calling \texttt{unify} (see
\cref{fig:sessionapi}) whose guard is \(r=r_0\), which is the dual of
\texttt{exify} since roles can only be 0 or 1 in a binary session. It
will turn the endpoint at the client side into \[
\forall m{:}\sort{int}.\forall n{:}\sort{int}.
\chan(\cli,\msg(\cli,\type{int}(m))\coloncolon\msg(\cli,\type{int}(n))\coloncolon\msg(\srv,\type{bool}(m=n))\coloncolon\nil(\srv))
\] Essentially, a universally quantified endpoint \emph{inputs} a static
term from the user to eliminate the quantifier, while an existentially
quantified endpoint \emph{outputs} the witness to the user to eliminate
the quantifier. Note that the user of an endpoint is the process holding
such endpoint as mentioned above. So ``inputs from the user'' means the
user writes a program to \emph{send} a value using the endpoint. Such a
twist is found in other works as well, e.g.
\cite{Wadler:2012ua,Wadler:2014bta}.

The main contribution of this paper lies in the formulation of a form of
dependent session types (of DML-style) in the setting of
\(\lambda\)-calculus, which is the first of its kind. In particular,
this formulation is based on unpolarized presentation. Our technical
results include preservation and progress properties, which guarantee
session fidelity and deadlock-freeness. We also mention at the end an
implementation of our system that targets Erlang/Elixir.

The rest of the paper is organised as follows. \cref{sec:mtlc-linear}
briefly sets up multi-threaded \(\lambda\)-calculus with linear types,
denoted as \(\Lzero\). \cref{sec:mtlc-linear-dep} introduces
\emph{predicatization} to extend \(\Lzero\) into multi-threaded
\(\lambda\)-calculus with dependent types and linear types, denoted as
\(\Ldep\). \cref{sec:session-types} further extends \(\Ldep\) to
formulate dependent session types as \(\Lchan\).
\cref{sec:implementations} describes technical details of our
implementations. \cref{sec:examples} demonstrates the benefits of
dependent session types through examples. We then mention extensions
(multi-party sessions, polymorphism, etc) in \cref{sec:ext}, related
works in \cref{sec:related-works} and finally conclude in
\cref{sec:conclusion}.

\section{\texorpdfstring{Multi-threaded \(\lambda\)-calculus with Linear
Types}{Multi-threaded \textbackslash{}lambda-calculus with Linear Types}}\label{sec:mtlc-linear}

The formulation of multi-threaded \(\lambda\)-calculus with linear types
is largely standard and follows exactly from our previous work
\cite{Xi:2016ue} except for some minor cosmetic changes. Therefore, we
only present it very briefly and refer the readers to our prior work for
details.

\subsection{Syntax}\label{syntax}

\input{mtlc0syntax} The syntax is shown in \cref{fig:mtlc0syntax} which
is mostly standard. \(\delta\)/\(\hat\delta\) are non-linear/linear base
types. ``vtype'' is just linear type. Note that a type \(\tau\) is also
a linear type \(\hat\tau\), but it is not regarded as a \emph{true}
linear type. \(\dcc\)/\(\dcf\) are dynamic constant
constructors/functions (pre-defined constructors/functions). \(\dcr\)
are dynamic constant resources that are treated \emph{linearly}.
\(\sig\) are dynamic signatures that assign types to dynamic constants,
and these types are called \emph{c-types}. Note that \(\vv{\cdot}\)
stands for a possibly empty sequence of \(\cdot\), i.e. \(\vv{e}\) is a
possibly empty sequence of dynamic terms. \(\dcx\,(\vv{e})\) is a term
of type \(\tau\) if \emph{dcx} is a constant of c-type
\((\tau_1,\dotsc,\tau_n)\Rightarrow\tau\) in \(\sig\) and for each
\(e_i~(1 \leq i \leq n)\) in \(\vv{e}\), \(e_i\) has type \(\tau_i\).

We use \([]\) for the empty mapping and
\([a_1,\dotsc,a_n \mapsto b_1,\dotsc,b_n]\) for a mapping that maps
\(a_i\) to \(b_i\) for \(1\leq i\leq n\), in which case we write
\(m(a_i)\) to mean \(b_i\). We use \(\dom(m)\) for the domain of a
mapping \(m\). If \(a \notin \dom(m)\), then \(m[a \mapsto b]\) means to
extend \(m\) with a new link from \(a\) to \(b\). We also use
\(m\backslash a\) to mean the mapping obtained by removing \(a\) from
\(\dom(m)\), and \(m[a\colonequals b]\) to mean
\((m\backslash a)[a \mapsto b]\). Substitution \(\theta\) is a mapping
from variables to dynamic values. We write \(e[\theta]\) for the result
of applying \(\theta\) to \(e\). Pool \(\Pi\) is a mapping from thread
identifiers \(t\) (represented as natural numbers) to closed dynamic
expressions such that \(0 \in \dom(\Pi)\). We use
\(\Pi(t), t \in \dom(\Pi)\) to refer to a thread in \(\Pi\) whose thread
identifier is \(t\). We use \(\Pi(0)\) for the main thread.

Typing contexts are divided into a non-linear part \(\Gamma\) and a
linear part \(\Delta\). They are intuitionistic meaning that it is
required that each variable occurs at most once in a non-linear context
\(\Gamma\) or a linear context \(\Delta\). Given \(\Gamma_1, \Gamma_2\)
s.t. \(\dom(\Gamma_1)\cap\dom(\Gamma_2)=\varnothing\), we write
\((\Gamma_1, \Gamma_2)\) for the union of the two. The same notion also
applies to linear context \(\Delta\). Given non-linear context
\(\Gamma\) and linear context \(\Delta\), we can form a combined context
\((\Gamma;\Delta)\) when \(\dom(\Gamma)\cap\dom(\Delta)=\varnothing\).
Given \((\Gamma;\Delta)\), we may write \((\Gamma;\Delta),x:\hat\tau\)
for either \((\Gamma;\Delta,x:\hat\tau)\) or
\((\Gamma,x:\hat\tau;\Delta)\) if \(\hat\tau\) is indeed a non-linear
type.

Besides integers and booleans, we also assume a constant function
\texttt{thread\_create} in \emph{dcx} whose c-type in \(\sig\) is
\((\type{1}\multimap\type{1})\Rightarrow\type{1}\). A function of type
\(\type{1}\multimap\type{1}\) takes no argument and returns no result
(if it terminates). Since it is a true linear function, it can be
invoked exactly once. Intuitively, \texttt{thread\_create} creates a
thread that evaluates the linear function. Its semantic is to be
formally introduced later.

To manage resources, we follow \cite{Xi:2016ue} and define
\(\rho(\cdot)\) (\cref{fig:rho}) to compute the multiset (bag) of
constant resources in a given expression and \(\res\) (\textbf{RES} in
\cite{Xi:2016ue}) to range over such multisets of resources. We say
\(R\) is valid if \(R\in\res\) holds. Intuitively, \(\res\) can be
thought as all the resources of all the programs and \(R\) the resources
of a single program. We need to make sure that resource allocation to
different programs is consistent in \(\res\). For precise definitions,
please refer to our prior work.

\subsection{Sementics}\label{sementics}

Typing rules are the same as \cite{Xi:2016ue}, and we push it to
\cref{fig:mtlc0typing} in the appendix. The c-type judgment based on the
signature is of the form \(\sig\vDash e:\hat\tau\). A typing judgment is
of the form \(\Gamma;\Delta\vdash e:\hat\tau\) which is standard. By
inspecting the rules in \cref{fig:mtlc0typing}, we can readily see that
a closed value cannot contain resources if it can be assigned a
non-linear type \(\tau\). The \emph{Lemma of Canonical Forms} and the
\emph{Lemma of Substitution} are the same as our previous work
(\cite{Xi:2016ue} Lemma 2.2 and Lemma 2.3), we thus omit them
completely.

\(\Lzero\) has a call-by-value semantic, and the definition of
evaluation context (\(E\)), redex, and reducts are completely standard
and are the same as our previous work. We thus omit the details and
present just reduction on pools and properties of \(\Lzero\). Given
pools \(\Pi_1,\Pi_2\), we define \emph{reductions on pools}
\(\Pi_1\rightarrow\Pi_2\) as follows,
\begin{gather*}
\prftree[r]{\bf pr0}
  {e_1\rightarrow e_2}
  {\Pi[t\mapsto e_1]\rightarrow\Pi[t\mapsto e_2]}
\quad
\prftree[r]{\bf pr2}
  {t > 0}
  {\Pi[t\mapsto\langle\rangle]\rightarrow\Pi}
\\
\prftree[r]{\bf pr1}
  {\Pi(t)=E[\texttt{thread\_create}(\dyn{lam}~x.e)]}
  {\Pi\rightarrow\Pi[t\colonequals E[\langle\rangle]][t'\mapsto\dyn{app}(\dyn{lam}~x.e, \langle\rangle)]}
\end{gather*}

\begin{theorem}[Subject Reduction on Pools]
\label{thm:subjectreduction}
Assume $\varnothing;\varnothing\vdash\Pi_1:\hat\tau$ is derivable and $\Pi_1\rightarrow\Pi_2$ holds for some $\Pi_2$ satisfying $\rho(\Pi_2)\in\res$. Then $\varnothing;\varnothing\vdash\Pi_2:\hat\tau$ is also derivable.
\end{theorem}

\begin{theorem}[Progress Property on Pools]
\label{thm:progress}
Assume that $\varnothing;\varnothing\vdash\Pi_1:\hat\tau$ is derivable. Then we have 
\begin{itemize}
\item $\Pi_1$ is a singleton mapping $[0\mapsto v]$ for some value $v$, or 
\item $\Pi_1\rightarrow\Pi_2$ holds for some $\Pi_2$ s.t. $\rho(\Pi_2)\in\res$.
\end{itemize}
\end{theorem}

\begin{theorem}[Soundness of $\Lzero$]
\label{thm:soundness}
Assume that $\varnothing;\varnothing\vdash\Pi_1:\hat\tau$ is derivable. Then for any $\Pi_2$, $\Pi_1\rightarrow^* \Pi_2$ implies that either $\Pi_2$ is a singleton mapping $[0\mapsto v]$ for some value $v$ or $\Pi_2\rightarrow\Pi_3$ for some $\Pi_3$ satisfying $\rho(\Pi_3)\in\res$, where $\rightarrow^*$ is the transitive and reflective closure of $\rightarrow$.
\end{theorem}\begin{proof}
Follows directly from \cref{thm:subjectreduction} and \cref{thm:progress}.
\end{proof}

\section{Predicatization}\label{sec:mtlc-linear-dep}

In this section, we extremely briefly describe an approach to extend
\(\Lzero\) to support both universally and existentially quantified
types. Such process is \emph{predicatization} and is mostly standard in
the framework of \ATS~\cite{Xi:2003kl}. Predicatization is extensively
described in \cite{Xi:1998wa,Xi:1999bh,Xi:2007te}, and has been employed
in several other papers based on \ATS, e.g.
\cite{Shi:2013hr,Shi:2008vp}. We thus only summarize the process to
prepare for the development of \(\Lchan\), and omit any technical
details.

As an applied type system, \(\Ldep\) is layered into \emph{statics} and
\emph{dynamics}. The dynamics of \(\Ldep\) is based on \(\Lzero\), while
the statics will be a newly introduced layer underlying \(\Lzero\). The
predicatization process concerns mostly about formalizing the type index
language while maintaining the dynamic semantics of \(\Lzero\), and
reducing type equality problems into constraint solving problems w.r.t.
some constraint domain, such as integer arithmetic. General steps of
predicatization involve the followings:
\begin{itemize}
\tightlist
\item
  Formalizing statics, the language of type index. This involves its
  syntax, sorting rules, and specifically, non-linear type/linear type
  formation rules, etc.
\item
  Formalizing type equality in terms of subtyping relations and regular
  constraint relations.
\item
  Extending dynamics. This involves extending the syntax, typings,
  evaluation context, and reduction relations to accommodate, for
  instance, the introduction and elimination of quantifiers.
\end{itemize}

\input{statics} \input{types} \input{mtlcdepsyntax}
\input{mtlcdeptyping}

The language of statics can be regarded as a simply typed
\(\lambda\)-calculus. The ``types'' for static terms are denoted as
\emph{sorts} to avoid confusion. The syntax for statics is shown in
\cref{fig:statics} which is mostly standard. We assume base sorts \(b\)
to include \(\sort{int}\), \(\sort{bool}\), \(\sort{type}\) for types,
and \(\sort{vtype}\) for linear types. Non-linear/linear types in the
\(\Ldep\) are now static terms of sorts
\(\sort{type}\)/\(\sort{vtype}\), respectively. We reformulate types in
the dynamics in \cref{fig:types}.

Given a proposition \(P\) (a static term of sort \(\sort{bool}\)) and a
type \(\tau\), \(P\guard\tau\) is a \emph{guarded type}, and
\(P\assert\tau\) is an \emph{asserting type}. Formal definition of
guarded types and asserting types can be found in \cite{Xi:2007te}.
Intuitively, in order to turn a value of type \(P\guard\tau\) into a
value of type \(\tau\), we must establish the proposition \(P\), thus
``guarded''; if a value of type \(P\assert\tau\) is generated, we can
assume that the proposition \(P\) holds, thus ``asserting''.

The extended syntax of \(\Ldep\) over that of \(\Lzero\) is given in
\cref{fig:mtlcdepsyntax}. Typing judgement in \(\Ldep\) is of the form
\(\Sigma;\vv{P};\Gamma;\Delta\vdash e:\hat\tau\) where \(\Sigma\) is the
sorting environment for static terms and \(\vv{P}\) is a sequence of
propositions keeping track of the constraints. We present only some
additional typing rules in \cref{fig:mtlcdeptyping}.

We claim that \cref{thm:subjectreduction}, \cref{thm:progress}, and
\cref{thm:soundness} can be carrier over to \(\Ldep\) following the
proof in \cite{Xi:2007te}.

\section{Dependent Session Types}\label{sec:session-types}

Dependent types are types that depend on terms, and they offer much more
expressive power for specifying intended behavior of a program through
types. A restricted form of dependent types, we call dependent types of
DML-style \cite{Xi:2007te}, are types that depend on \emph{static}
terms. In this section, we will formally develop \emph{dependent session
types} (of DML-style), where session types can have quantification over
static terms. Based on \(\Ldep\), we first extend the statics, then
extend the dynamics, and finally discuss the soundness of \(\Lchan\).

\subsection{Extending Statics}\label{extending-statics}

The syntax of extended statics is given in \cref{fig:sessionsyntax}. We
add \(\sort{stype}\) as a new base sort to represnet session types.
Session types \(\pi\) are now static terms of sort \(\sort{stype}\). We
use \(i\) for static integers and \(b\) for static booleans. \(\nil(i)\)
means party \(i\) (the party holding endpoint \(i\)) will close the
session while the other party will wait for closing. Given linear type
\(\hat\tau\) and a session type \(\pi\),
\(\msg(i,\hat\tau)\coloncolon\pi\) means party \(i\) should send a
message to the other party, and then continue as \(\pi\).
\(\branch(i,\pi_1,\pi_2)\) is for branching, where party \(i\) should
choose to continue as \(\pi_1\) or \(\pi_2\) while the other party
simply follows the choice. Beyond these basic session type constructs,
we have \(\ite\)\footnote{Note that \(\branch\) is just a special case
  of \(\ite\) and we can indeed encode \(\branch\) using \(\ite\).} for
conditional branch, \(\quan\) for universal/existential quantification,
and \(\fix\) for recursions. Given a static boolean expression,
\(\ite(b, \pi_1,\pi_2)\) represents \(\pi_1\) when \(b\) is \(\top\)
(true), or \(\pi_2\) when \(b\) is \(\bot\) (false). Given a static
function of sort \(\sigma\rightarrow\sort{stype}\),
\(\quan(i,\lambda a{:}\sigma.\pi)\) is interpreted
\emph{intuitively}\footnote{This is only intuitively interpreted. Its
  accurate interpretation should be considered together with an endpoint
  since \(\pi\) is global. See later sections.} as universally
quantified \(\forall a{:}\sigma.\pi\) by party \(i\), or as
existentially quantified \(\exists a{:}\sigma.\pi\) by the other party.
Note that this is actually a session type \emph{scheme} and we assume
the existance of such \(\quan\) for every sort \(\sigma\). The need for
a unified representation of quantifiers, \texttt{quan}, is a must since
we essentially formulate all session types as \emph{global}, as compared
to polarized presentation where session types are all \emph{local}.
Given a static function of sort \(\sort{stype}\rightarrow\sort{stype}\),
\(\fix(\lambda a{:}\sort{stype}.\pi)\) is an encoding of the fixpoint
operator that represents the fixpoint of the input function. In
practice, we may write recursive definitions directly as a syntax sugar
(as shown in \cref{ex:array}).

\input{sessionsyntax}

Besides, we also introduce \emph{role} as a subset sort
\(\{r{:}\sort{int}\mid r=0\lor r=1\}\) to represent two parties, server
(0) and client (1), involved in a binary session. Note that subset sorts
are merely syntax sugars for a guarded/asserting type \cite{Xi:1999bh}.
For instance, \(\forall r{:}\sort{role}.\type{int}(r)\) is desugared
into \(\forall r{:}\sort{int}.(r=0\lor r=1)\guard\type{int}(r)\). We
also add the following \emph{linear} type constructor as a static
constant\footnote{It is indeed
  \(\chan:(\sort{int},\sort{stype})\Rightarrow\sort{vtype}\) since in
  \ATS, subset sort is not allowed in a c-sort. We use \(\sort{role}\)
  here just to simplify our presentation.},
\[\chan : (\sort{role},\sort{stype})\Rightarrow\sort{vtype}\] that
represents a linear channel. Given role \(r\) and session type \(\pi\),
\(\chan(r,\pi)\) is endpoint \(r\) of a channel held by a party. The
channel is governed by the session type \(\pi\), and the endpoint
interprets this session type \emph{locally} as role \(r\).

\subsection{Extending Dynamics}\label{extending-dynamics}

We add the following dynamic constant functions (pre-defined functions),
shown in \cref{fig:sessionapi}, to create, use, and consume linear
channels. We will refer to them as \emph{session API} or just the API.
We break up the figure and present them with explanations here. \[
\texttt{create}  : \forall r_1,r_2{:}\sort{role}.\forall\pi{:}\sort{stype}.(r_1\neq r_2)\guard(\type{chan}(r_2,\pi)\multimap\type{1})\Rightarrow\type{chan}(r_1,\pi) 
\] \texttt{create} is to create a session of two threads, connected via
a channel of session type \(\pi\), and each thread holds an endpoint of
the channel. One party is holding endpoint \(r_1\) of type
\(\chan(r_1,\pi)\) as returned by \texttt{create} in the current thread,
while the other party is holding endpoint \(r_2(\neq r_1)\) of type
\(\chan(r_2,\pi)\) in a newly spawned thread evaluating the given
\emph{linear} function of type \(\chan(r_2,\pi)\multimap\type{1}\). As
the (closure) function may contains resources, it must be linear to
guarantee that it can be called \emph{exactly once}. The channel
endpoint will be consumed in this function as it is linear.
\begin{align*}
\texttt{send}    &: \forall r,r_0{:}\sort{role}.\forall\pi{:}\sort{stype}.\forall \hat\tau{:}\sort{vtype}.(r=r_0)\guard(\chan(r,\msg(r_0,\hat\tau)\coloncolon\pi), \hat\tau)\Rightarrow\chan(r,\pi) 
\\
\texttt{recv}    &: \forall r,r_0{:}\sort{role}.\forall\pi{:}\sort{stype}.\forall \hat\tau{:}\sort{vtype}.(r\neq r_0)\guard\chan(r,\msg(r_0,\hat\tau)\coloncolon\pi)\Rightarrow \hat\tau\otimes\chan(r,\pi) 
\end{align*}
\texttt{send} is for sending linear values. Given global session type
\(\msg(r_0,\hat\tau)\coloncolon\pi\), its interpretation at \(r\) where
\(r=r_0\) is to send a message of linear type \(\hat\tau\) then to
proceed as \(\pi\). The \texttt{send} function \emph{consumes} the
channel, uses the capability of sending denoted by
\(\msg(r_0,\hat\tau)\), and returns another channel of type
\(\chan(r,\pi)\), where the sending capability is now removed. Dually,
the interpretation of \(\msg(r_0,\hat\tau)\coloncolon\pi\) is to receive
at party \(r(\neq r_0)\), implemented by \texttt{recv}. Note that even
though we encode it here in the style of continuation, our
implementation directly \emph{changes} the type of channel without
consuming it. In ATS programming language, it is presented in the
following style,
\begin{multline*}
\texttt{send}:\forall r,r_0{:}\sort{role}.\forall\pi{:}\sort{stype}.\forall \hat\tau{:}\sort{vtype}.\\
(r=r_0)\guard(!\chan(r,\msg(r_0,\hat\tau)\coloncolon\pi)\gg\chan(r,\pi), \hat\tau)\Rightarrow\type{1}
\end{multline*}
Similarly, \texttt{close} is for terminating a session while
\texttt{wait} is waiting for the other side to close.
\begin{align*}
\texttt{close}   &: \forall r,r_0{:}\sort{role}.(r=r_0)\guard\chan(r,\stype{end}(r_0))\Rightarrow\type{1} 
\\
\texttt{wait}    &: \forall r,r_0{:}\sort{role}.(r\neq r_0)\guard\chan(r,\stype{end}(r_0))\Rightarrow\type{1} 
\end{align*}
The interpretation of \(\branch(r_0,\pi_1,\pi_2)\) at party
\(r(\neq r_0)\) is to offer two choices, \(\pi_1\) and \(\pi_2\).
Therefore, \texttt{offer} function will consume the endpoint and return
a linear pair of the other party's choice (as a singleton boolean) and
the endpoint whose session type is a conditional branch between
\(\pi_1,\pi_2\) using the received tag \(b\) as the condition. Dually,
\texttt{choose} will choose \(\pi_1\) and \(\pi_2\) respectively
according to the boolean tag provided by the user. Note that these two
functions are completely unnecessary since they can be encoded using
other functions/session types. We present them here just to stay inline
with others where \texttt{offer/choose} are usually treated as standard
constructs.
\begin{align*}
\texttt{offer}   &: \forall r,r_0{:}\sort{role}.\forall \pi_1,\pi_2{:}\sort{stype}.(r\neq r_0)\guard\chan(r,\stype{branch}(r_0,\pi_1,\pi_2)) \\
                               &\qquad\qquad\Rightarrow\exists b{:}\sort{bool}.\type{bool}(b)\otimes\chan(r,\ite(b,\pi_1,\pi_2))
\\
\texttt{choose} &: \forall r,r_0{:}\sort{role}.\forall \pi_1,\pi_2{:}\sort{stype}.\forall b{:}\sort{bool}.(r=r_0)\guard(\chan(r,\stype{branch}(r_0,\pi_1,\pi_2)), \type{bool}(b)) \\
                               &\qquad\qquad\Rightarrow \chan(r,\ite(b,\pi_1,\pi_2)) 
\end{align*}
\texttt{unify} is to interpret \(\quan(r_0,\cdot)\) at party
\(r(= r_0)\) as universal quantifier, while \texttt{exify} is to
interpret it dually as existential quantifier at party \(r(\neq r_0)\).
\begin{align*}
\texttt{unify}   &: \forall r,r_0{:}\sort{role}.\forall\pi{:}\sort{stype}.\forall f{:}\sigma\rightarrow\sort{stype}.
                              \\&\qquad\qquad (r= r_0)\guard\chan(r,\quan(r_0,f))\Rightarrow\forall s{:}\sigma.\chan(r,f(s)) \\
\texttt{exify}   &: \forall r,r_0{:}\sort{role}.\forall\pi{:}\sort{stype}.\forall f{:}\sigma\rightarrow\sort{stype}.
                              \\&\qquad\qquad (r\neq r_0)\guard\chan(r,\quan(r_0,f))\Rightarrow\exists s{:}\sigma.\chan(r,f(s)) 
\end{align*}
\texttt{itet} and \texttt{itef} reduces the conditional branching
session type \(\ite(b,\pi_1,\pi_2)\) according to static boolean
expression \(b\). \texttt{recurse} unrolls the fixpoint encoding.
\begin{align*}
\texttt{itet}    &: \forall r{:}\sort{role}.\forall \pi_1,\pi_2{:}\sort{stype}.\chan(r,\stype{ite}(\top,\pi_1,\pi_2))\Rightarrow\chan(r,\pi_1)\\
\texttt{itef}    &: \forall r{:}\sort{role}.\forall \pi_1,\pi_2{:}\sort{stype}.\chan(r,\stype{ite}(\bot,\pi_1,\pi_2))\Rightarrow\chan(r,\pi_2)\\
\texttt{recurse} &: \forall r{:}\sort{role}.\forall f{:}\sort{stype}\rightarrow\sort{stype}.\chan(r,\fix(f))\Rightarrow\chan(r,f(\fix(f))) 
\end{align*}
Note that these functions
(\texttt{unify}/\texttt{exify}/\texttt{itet}/\texttt{itef}/\texttt{recurse})
are \emph{proof} functions that merely change the types of endpoints.
They have no runtime counterparts and thus can be eliminated after type
checking has passed.

\emph{Duality} is not explicitly encoded as is usually done in session
types literature \cite{Lindley:2016du,Pucella:2008dt,Jespersen:2015ka}.
Instead, we choose to make the duality as general as possible and use a
\emph{global} session type \(\pi\) paired with a role \(r\) to guide the
local interpretation at endpoint \(r\). Given that \(r\) can only be
\(0\) or \(1\), we can define that \(\chan(0,\pi)\) and \(\chan(1,\pi)\)
are \emph{dual} endpoints of a channel. Session API come in dual pairs,
and the dual usage of dual endpoints are realized by the corresponding
session API pairs with the help of guarded types. The typing rules for
guarded types will force one endpoint to be only used with one API in
the pair while the dual endpoint to be only used with the dual API in
the same pair. A crucial indication of such formulation is that we
essentially reduce the duality checking problem into a simple integer
comparison problem, which greatly simplifies our formulation. Also, it
reduces the number of the dynamic constants in \cref{fig:sessionapi} in
half by avoiding coercion between so-called input/output types
\cite{Lindley:2016du}. In our previous work \cite{Xi:2016ue}, we used a
polarized presentation, e.g. \(\type{chanpos}(p)\) and
\(\type{channeg}(p)\) where \(p\) is a \emph{local} type. This is
similar to \texttt{In{[}{]}}/\texttt{Out{[}{]}} in \cite{Scalas:2016uh},
\(S_?\)/\(S_!\) in \cite{Lindley:2016du} Section 6, and
\emph{dual}/\emph{notDual} in \cite{Sackman:2008ux}. We found this
polarized presentation is not suitable for extending to multi-party
sessions, whereas our ``global+role+guard'' formulation can be very
easily adapted to multi-party sessions based on \cite{Xi:2017wv}. For
example, in a three-party session, we can define \(\chan(0,\pi)\),
\(\chan(1,\pi)\), and \(\chan(2,\pi)\) to be \emph{compatible}, as a
generalization to duality. We very briefly mention such extension in
\cref{sec:ext}. \[
\texttt{cut}     : \forall r_1,r_2{:}\sort{role}.\forall\pi{:}\sort{stype}.(r_1\neq r_2)\guard(\chan(r_1,\pi),\chan(r_2,\pi))\Rightarrow\type{1}
\] Given \emph{dual} endpoints, \texttt{cut} will link together the
endpoints by performing \emph{bi-directional forwarding}. In other
words, it will send onto one endpoint each received value from the other
endpoint. \texttt{cut} is often used to implement delegation of service.
It can be proven that these two endpoints must belong to
\emph{different} channels since otherwise, it will obviously deadlock.
We will explain more in \cref{sec:implementations}.

\subsection{Dynamic Semantics}\label{dynamic-semantics}

The dynamic semantics of \(\Lchan\) is indeed the same as our prior work
except that we have added a branching construct and we use a more
general unpolarized presentation. We thus push additional reduction ruls
on pools in \cref{fig:sessionreduction} and
\cref{fig:sessionreductioncut} to the appendix. Note that, as mentioned
above,
\texttt{unify}/\texttt{exify}/\texttt{itet}/\texttt{itef}/\texttt{recurse}
do not have any dynamic semantics. The meaning of these rules should be
intuitively clear. For instance, \textbf{pr-msg} states, if thread
\(t_1\) in pool \(\Pi\) is of the form
\(E[\texttt{send}(ch_{i,r_1},v)]\), and thread \(t_2\) in pool \(\Pi\)
is of the form \(E[\texttt{recv}(ch_{i,r_2})]\), then \(\Pi\) can be
reduced to another pool where \(t_1\) is replaced by \(E[ch_{i,r_1}]\)
and \(t_2\) is replaced by \(E[\langle v,ch_{i,r_2}\rangle]\).

\subsection{Soundness of the Type
System}\label{soundness-of-the-type-system}

While \cref{thm:subjectreduction} can be easily established for
\(\Lchan\), \cref{thm:progress} is more involved due to the addition of
session API. However, based on \cite{Xi:2003kl,Xi:1999bh}, \(\Ldep\) and
\(\Lchan\) are \emph{conservative} extensions of \(\Lzero\), and the
deadlock-freeness is proven for \(\Lzero\) with channels in
\cite{Xi:2016ue} using a technique known as \emph{DF-Reducibility}. Thus
the same results can be proven for \(\Lchan\) using the exact same
technique since the dynamic semantics are the same. We thus refer
readers to \cite{Xi:2016ue,Xi:1999bh} for detailed proofs. We can then
establish the same deadlock-freeness guarantee as stated in Lemma 3.1 of
\cite{Xi:2016ue}

\begin{theorem}[Subject Reduction of $\Lchan$]
Assume that $\varnothing;\varnothing;\varnothing;\varnothing\vdash\Pi_1:\hat\tau$ is derivable and $\Pi_1\rightarrow\Pi_2$ s.t. $\rho(\Pi_2)\in\res$. Then $\varnothing;\varnothing;\varnothing;\varnothing\vdash\Pi_2:\hat\tau$ is also derivable.
\end{theorem}

\begin{theorem}[Progress Property of $\Lchan$]
Assume that $\varnothing;\varnothing;\varnothing;\varnothing\vdash\Pi_1:\hat\tau$ is derivable and $\rho(v)$ contains no channel endpoins for every $v:\hat\tau$. Then 
\begin{itemize}
\item $\Pi_1$ is a singleton mapping $[0\mapsto v]$ for some $v$, or 
\item $\Pi_1\rightarrow\Pi_2$ holds for some $\Pi_2$ s.t. $\rho(\Pi_2)\in\res$.
\end{itemize}
\end{theorem}

\begin{theorem}[Soundness of $\Lchan$]
Assume that $\varnothing;\varnothing;\varnothing;\varnothing\vdash\Pi_1:\hat\tau$ is derivable and $\rho(v)$ contains no channel endpoins for every $v:\hat\tau$. Then for any $\Pi_2$ satisfying $\rho(\Pi_2)\in\res$, $\Pi_1\rightarrow^*\Pi_2$ implies either $\Pi_2$ is a singleton mapping $[0\mapsto v]$ for some $v$, or $\Pi_2\rightarrow\Pi_3$ for some $\Pi_3$ s.t. $\rho(\Pi_3)\in\res$.
\end{theorem}

\section{Implementations}\label{sec:implementations}

Our implementations consist of two parts, a session API library in ATS,
and a runtime implementation of the session API (referred to as a
\emph{back-end}) in a target language. ATS is a programming language
based on \ATS, and it supports a style of \emph{co-programming} with
many target languages by compiling an ATS program into the target
language. Its default compilation target is C. For the purpose of this
paper, besides a native back-end in ATS/C itself, we also support
back-ends in Erlang/Elixir and JavaScript. A session-typed program will
be firstly type-checked based on the type system of \(\Lchan\), and then
compiled into a target language (if passed type checking). The
compiler/interpreter of the target language will then be invoked to
compile/interpret the program together with the corresponding back-end.
Although formalized as synchronous sessions (for the sake of
simplicity), our implementations can fully support asynchronous
communications. Our linear typing guarantees \emph{no resources leaks}.
For instance, in our Erlang/Elixir back-end, there are no process leaks
related to channels.

Our session API library in ATS is (almost) a direct translation of those
listed in \cref{fig:sessionapi}, except for some slight syntax
differences. For example, \texttt{send} is translated into the
followings.

\begin{Shaded}
\begin{Highlighting}[]
\KeywordTok{fun}\NormalTok{ send }\DataTypeTok{\{r,r0:role|r0==r\}} \DataTypeTok{\{p:stype\}} \DataTypeTok{\{v:vtype\}} 
\NormalTok{         (!chan(r,msg(r0,v)::p) >> chan(r,p), v): void}
\end{Highlighting}
\end{Shaded}
where \texttt{\{\}} is universal quantification (and \texttt{{[}{]}} is
existential quantification), \texttt{!} means call-by-value, which
indicates \emph{not} to consume a linear value, and
\texttt{\textgreater{}\textgreater{}} means to \emph{change} the linear
type after the function returns. As mentioned before, whenever possible,
the API will change the types of endpoints directly instead of relying
on continuations. There are a couple other minor changes. First, with
guarded recursive data types \cite{Xi:2003dp} and pattern matching, the
API formulates \texttt{offer}/\texttt{choose} in a simpler way as
follows,

\begin{Shaded}
\begin{Highlighting}[]
\KeywordTok{datatype}\NormalTok{ choice (stype, stype, stype) = }
\NormalTok{| \{p,q:stype\} Left  (p, p, q) }\KeywordTok{of}\NormalTok{ ()}
\NormalTok{| \{p,q:stype\} Right (q, p, q) }\KeywordTok{of}\NormalTok{ ()}
\KeywordTok{fun}\NormalTok{ offer  }\DataTypeTok{\{r,r0:role|r0!=r\}} \DataTypeTok{\{p,q:stype\}} 
\NormalTok{           (!chan(r,branch(r0,p,q)) >> chan(r,s)): }\DataTypeTok{[s:stype]}\NormalTok{ choice (s,p,q)}
\KeywordTok{fun}\NormalTok{ choose }\DataTypeTok{\{r,r0:role|r0==r\}} \DataTypeTok{\{p,q,s:stype\}} 
\NormalTok{           (!chan(r,branch(r0,p,q)) >> chan(r,s), choice(s,p,q)): void}
\end{Highlighting}
\end{Shaded}
where \texttt{choice} is a guarded recursive data type that essentially
captures the equality on session types. Also, since it is existentially
quantified, the type-checker will enforce \emph{exhuastive} case
analysis on the received choice to instantiate \texttt{s}. Note that
\texttt{s} as in \texttt{\textgreater{}\textgreater{}\ chan(r,s)} is in
the scope of quantifier \texttt{{[}s:stype{]}} even though it appears
before the quantifier.

We briefly mention some technical details below and refer the readers to
\url{http://multirolelogic.org} for pointers to all the source code. Due
to space limitation, we assume that the readers are reasonably familiar
with these target languages.

\subsection{Message-passing Back-end in
Erlang/Elixir}\label{message-passing-back-end-in-erlangelixir}

Erlang offers functional distributed programming abilities through its
powerful virtual machine. Elixir offers a more friendly syntax and
better tooling on top of the same runtime. In Erlang/Elixir, every
process has a unique \texttt{pid} (process identifier), and an
associated mailbox. Communications are achieved via message-passing
asynchronously and can be done across different nodes. In this
particular implementation, \texttt{choose} and \texttt{offer} are
implemented as \texttt{send} and \texttt{receive}, respectively.
\texttt{close} and \texttt{wait} are implemented both to terminate the
process directly. This back-end relies on order-preserving messages and
is inherently asynchronous and distributed.

In Erlang/Elixir back-end, a message is represented by a label, a
\texttt{pid}, a \texttt{ref}, and a payload. A channel endpoint is
identified through a combination of a \texttt{pid} and a \texttt{ref}.
The message labels are used to identify the kind of messages, e.g.
\texttt{:send}/\texttt{:receive}. The \texttt{pid} is used to locate the
message's origin, or an endpoint's mailbox. The \texttt{ref}'s are
globally unique references, generated through a built-in function
\texttt{make\_ref} for every endpoint. The need for \texttt{ref} is
discussed in \cite{Mostrous:2011vu}. Intuitively speaking, the
\texttt{ref} acts as a signature of the message and every out-going
message is signed using the sending endpoint's own \texttt{ref}. Thus it
can be used both to distinguish in-session messages from out-of-session
messages\footnote{This is because that knowing just the \texttt{pid} is
  enough for any process to randomly inject messages to its mailbox.},
and to identify requests from the endpoint's owning process and messages
from the dual endpoint.

An endpoint will run a loop in a dedicated process and talk to the
owning process through messages-passing. The endpoint loop keeps track
of two parameters: \texttt{self}, which is its own signature as a
\texttt{ref}, and \texttt{dual}, which is the dual endpoint's
\texttt{pid} and \texttt{ref}. In every iteration, the loop will receive
a request from the owning process by pattern matching against messages
signed by \texttt{self}, and then process the request accordingly. For
instance, when the owning process sends a message with label
\texttt{:receive} signed with \texttt{self}, the endpoint will then
pattern match against messages in the endpoint's mailbox and block until
it finds the first message whose label is \texttt{:send} and is signed
by the dual endpoint's \texttt{ref}, which is \texttt{dual.ref}. The
found message will then be delivered to the owning process's mailbox,
fulfilling the request.

\input{elixirdelegate}

\texttt{cut} is implemented as delegation, where \texttt{:send} requests
are handled as before, but \texttt{:receive} requests are delegated to
an endpoint involved in a \texttt{cut}. Suppose we have dual endpoints
\texttt{A:chan(0,p)}/\texttt{A\textquotesingle{}:chan(1,p)} and dual
endpoints \texttt{B\textquotesingle{}:chan(0,p)}/\texttt{B:chan(1,p)} of
some session type \texttt{p}, and we are to perform
\texttt{cut(A\textquotesingle{},B\textquotesingle{})}. The owning
process \(P_2\) of both \texttt{A\textquotesingle{}} and
\texttt{B\textquotesingle{}}, will send a \texttt{:cut} request to
\texttt{A\textquotesingle{}} and \texttt{B\textquotesingle{}}, with a
payload of the \texttt{pid} and \texttt{ref} of
\texttt{B\textquotesingle{}} and \texttt{A\textquotesingle{}},
respectively. The info about \texttt{B\textquotesingle{}} will be
forwarded to \texttt{A}, and \texttt{A} will delegate \texttt{:receive}
requests to \texttt{B\textquotesingle{}}. Similarly, the info about
\texttt{A\textquotesingle{}} will be forwarded to \texttt{B}. and
\texttt{B} will delegate \texttt{:receive} requests to
\texttt{A\textquotesingle{}}. A delegated request will change its
signature from the original requester's \texttt{ref}, to the delegator's
\texttt{ref}, so that the delegator can still process the request as if
the request comes from its owning process. An example is illustrated in
\cref{fig:elixir}, where \(\leftrightarrow\) is for endpoint ownership,
\(\Leftrightarrow\) connects dual endpoints, and dashed arrow denotes
delegation. Now, if \(P_1\) sends a message to \(P_3\), it will be sent
through endpoint \texttt{A}, and then delivered to the mailbox of
\texttt{A\textquotesingle{}}. When \(P_3\) tries to receive the message,
it will send a \texttt{:receive} request to \texttt{B}, and \texttt{B}
delegates it to \texttt{A\textquotesingle{}}, and
\texttt{A\textquotesingle{}} will fulfill the request since the message
is in its mailbox.

We also have a shared memory implementation in ATS/C which implements
our own message queue guarded by locks, and a continuation-based
implementation in JavaScript using WebWorker.

\section{Examples}\label{sec:examples}

We will show some example dependent session types or programs in the
followings. We will assume that the server plays role 0 (\texttt{S}),
and the client plays role 1 (\texttt{C}). We will use ATS's ML-like
syntax to present the program (after omitting some insignificant
details), which can be easily mapped to \(\Lchan\). We also use syntax
sugar and implementation optimizations described in
\cref{sec:implementations} and extensions from \cref{sec:ext}. Again,
the source code can be found online through
\url{http://multirolelogic.org}, and all the code can be type-checked,
compiled, and executed.

\begin{example}[Counter]
One can easily define a counter as an integer stream. But more precisely, we can define dependently session typed constructor \stype{counter} as
\end{example}
\[
\stype{counter} (n{:}\sort{int}) \defeq \branch(\cli, \msg(\srv, \type{int}(n)) \coloncolon \stype{counter}(n+1),\nil(\cli)) 
\] which says, in every iteration, the client can choose to receive an
integer \(n\) and let the session continue from \(n+1\), or to end the
session. \texttt{counter} makes use of higher-order fixpoint encoding,
\texttt{fix}, which is better explained in \cref{ex:array}. On top of
\texttt{counter}, we can define a service \texttt{from} that given an
integer \(n\), returns an endpoint of session type
\texttt{counter}\((n)\). \[
\stype{from} \defeq \quan(\cli, \lambda n{:}int.\msg(\cli,\type{int}(n))\coloncolon\msg(\srv,\type{chan}(\cli,\stype{counter}(n)))\coloncolon\nil(\cli))
\] Since \(\type{chan}\) is a linear type constructor, a channel can
then be sent over another channel just as other linear values, and
\texttt{send} will consume it. This forms a higher-order session type.
We omit any testing code since it is similar to \cref{ex:array}. Due to
space limitation, we push other examples to \cref{app:ex}.

\section{Extensions}\label{sec:ext}

We very briefly describe possible extensions of \(\Lchan\). First, it is
straightforward to add \emph{general recursion} to our language (not to
the session type) as has been done in \cite{Xi:2016ue}. Second, one can
always introduce a \emph{higher-order} \texttt{fix} into session types,
such as
\[\stype{fix}(\lambda f{:}(\vv\sigma\rightarrow\sort{stype}).\lambda \vv a{:}\vv\sigma.\pi),\vv s)\]
where \(f\) is a static function of sort
\((\vv\sigma\rightarrow\sort{stype})\rightarrow\vv\sigma\rightarrow\sort{stype}\),
and \(\vv s\) are static terms of matching sorts \(\vv\sigma\).
Correspondingly, we need to introduce another \texttt{recurse} to unroll
it. A higher-order \texttt{fix} will input static terms to form a new
session type that dependents on these static terms. Thus these are also
a form of dependent session types. Third, binary branching can be
extended as well. For instance, we can introduce
\(\branch(i,\pi_1,\pi_2,\pi_3), i\in\{0,1,2\}\) and cooresponding
session API similar to \texttt{ite} to unroll it.

More importantly, we can extend \(\Lchan\) to support \emph{multi-party
session types} based on \cite{Xi:2017wv}. Roles will be extended from
\(\{0,1\}\) to a larger set of natural numbers, \(\chan(r,\pi)\) will be
extended to \(\chan(R,\pi)\) where \(R\) is now a \emph{set} of roles.
This is essential because of the need to represent one party's
\emph{complement} roles, which has to be a set. Guards in session API
will change from \(r=r_0\) to \(r_0\in R\), and from \(r\neq r_0\) to
\(r_0\notin R\). \texttt{cut} will be extended to another form based on
\cite{Xi:2017wv}.

Also, both predicative quantification (dependent types) and
higher-order/impredicative quantification (polymorphism) are supported
by \ATS, and our formulation naturally supports \emph{polymorphic
session types} in the sense of \cite{Caires:2013jb} since \texttt{quan}
and higher-order \texttt{fix} can input session types to form a session
type. We give such an example in \cref{ex:poly}. However, we focus on
dependent session types in this paper.

\section{Related Works}\label{sec:related-works}

To our best knowledge, \cite{Toninho:2011in} is the only other
formalization of dependent session types in a similar sense as ours. It
is based on intuitionistic linear type theory for a variant of
\(\pi\)-calculus, which extends the work in \cite{Caires:2010gi} where a
kind of Curry-Howard isomorphism is established. The work concerns with
two layers, an unspecified dependently typed layer for functional terms
that assign meanings to atomic propositions, and a session typed layer
that composes sessions and interprets linear logic connectives.
Quantifiers connect these two layers where universal quantifier inputs a
functional term and existential quantifier outputs a functional term.
Their line of works presents session types in a polarized style,
corresponding to the left/right introduction/elimination rules of the
logic. Our work is different in many ways. Our work is based on
\(\lambda\)-calculus instead of \(\pi\)-calculus/linear logic, and we
have shown our concrete implementations to support the argument that
such formulation is practical. Quantifiers are handled slightly
differently. We present unpolarized global quantifiers in the session
type, then locally interpreted it as \(\forall\)/\(\exists\) through our
session API. However, the input/output action is not limited to follow
the quantifiers immediately as they do. Our unpolarized style is easier
to extend to multi-party sessions, while theirs is inherently binary due
to the nature of duality in the logic. \cite{Caires:2013jb} and
\cite{Pfenning:2011ce} are based on \cite{Toninho:2011in} which focus on
polymorphic session types and proof-carrying code in session types,
respectively. Our work supports polymorphic session types in the sense
of \cite{Caires:2013jb} but we do not have space to formally address it.

There are many attempts to integrate session types into practical
programming languages.
\cite{Pucella:2008dt,Lindley:2016du,Sackman:2008ux} embed session types
into Haskell, \cite{Scalas:2016uh} in Scala, \cite{Jespersen:2015ka} in
Rust, \cite{Ng:2012hu} in C, and \cite{Hu:2008hta,Ng:2011go,Hu:2010fj}
in Java. The single sailent feature is that we support dependent session
types while none of above supports. Our type system also guarantees
linearity and duality natively and staticly without any special
encoding. Due to the lack of linear types, \cite{Lindley:2016du} relies
on an encoding of linear \(\lambda\)-calculus,
\cite{Pucella:2008dt,Sackman:2008ux} rely on indexed monads.
\cite{Jespersen:2015ka} makes use of affine types in Rust that
guarantees ``at most once'' usage which is still not enough. Other works
did not capture linearity in the type system. Duality is encoded as a
proof system using type classes in \cite{Pucella:2008dt,Lindley:2016du},
and using traits in \cite{Jespersen:2015ka}. \cite{Scalas:2016uh} uses
Scala's \texttt{In{[}-{]}}/\texttt{Out{[}-{]}} types where \texttt{-} is
a \emph{local} type, and similarly \cite{Sackman:2008ux} uses
\texttt{dual/notDual}, and they are both similar to our prior work using
\textbf{chanpos} and \textbf{channeg}. \cite{Hu:2008hta} ensures duality
in the runtime and \cite{Ng:2011go,Hu:2010fj} are its extensions. There
are other works proposing new languages to support session types, such
as \cite{Toninho:2013iu,Gay:2010gt,Wadler:2012ua} and
\texttt{SILL}\footnote{\url{https://github.com/ISANobody/sill}}
\cite{Caires:2010gi}, but these are not as practical in their current
states.

There are other works that are loosely related to ours, such as those
investigating links between logics and session types
\cite{Wadler:2012ua,Wadler:2014bta,Caires:2010gi}. Please refer to
\cite{Xi:2016ue} for more due to space limitations.

\section{Conclusion}\label{sec:conclusion}

We have presented a form of dependent session type system \(\Lchan\)
based on \(\lambda\)-calculus using unpolarized presentation. Our type
system handles quantification over static terms in session types,
allowing more precise session protocols to be described elegantly.
Linearity is guaranteed statically by the type system, and duality is
guaranteed by a combination of global session types, roles at a local
endpoint, and guards in the session API. \(\Lchan\) also supports
delegations, higher-order sessions, polymorphic sessions, and
recursively defined sessions. Our type system enjoys subject reduction
and progress properties, which guarantees session fidelity and
deadlock-freeness. We have shown the practicality of \(\Lchan\) by
providing a back-end in Erlang/Elixir, which is asynchronous,
distributed, and leak-free. Our formulation can also be adapted to
multi-party sessions based on multirole logic and we leave this as a
future work.

\vfill
\appendix

\section{Appendix - More Examples}\label{app:ex}

\begin{example}[Array]\label[ex]{ex:array}
One can safely send an array by sending a length $n$ first, then followed by $n$ messages for $n$ elements of the array. Such a channel can be encoded in the following dependent session types.
\begin{align*}
\stype{repeat}  (\tau{:}\sort{type}, n{:}\sort{int}) &\defeq \stype{ite}(n>0, \msg(\srv, \tau) \coloncolon \stype{repeat}(\tau,n-1),\nil(\srv)) \\
\stype{array}   (\tau{:}\sort{type}) &\defeq \quan(\srv, \lambda n{:}\sort{int}.\msg(\srv, \type{int}(n))::\stype{repeat}(\tau, n))
\end{align*}
\end{example}
where \texttt{repeat} is a recursive session type constructor written in
direct style, and its desugared version is as follows,
\begin{multline*}
\stype{repeat}  (\tau{:}\sort{type}, n{:}\sort{int}) \defeq \\
\fix(
  \lambda p{:}\sort{int}\rightarrow\sort{stype}.
    \lambda n{:}\sort{int}.
      \ite(n>0, \msg(\srv, \tau) \coloncolon p(n-1), \nil(\srv)), n) 
\end{multline*}
Note that \texttt{repeat} and \texttt{array} are session type
constructors, which are just static functions returning static terms of
sort \(\sort{stype}\). Also, the \texttt{fix} is a higher-order fixpoint
described in \cref{sec:ext}. \texttt{repeat}\((\tau,n)\) then says, if
\(n>0\) is true, the session proceeds to allow sending of a value of
type \(\tau\) from party \texttt{S} (\(\msg(\srv,\tau)\)), then proceeds
as \texttt{repeat}\((\tau,n-1)\). If \(n>0\) is false, the session can
only be terminated by party \texttt{S} (\(\nil(\srv)\)). Similarly,
\texttt{array} says, party \texttt{S} is to send an integer \(n\)
followed by \(n\) repeated messages described by
\texttt{repeat}\((\tau,n)\). Therefore, the server side can be
programmed as follows,

\begin{Shaded}
\begin{Highlighting}[]
\KeywordTok{fun}\NormalTok{ server }\DataTypeTok{\{a:type\}} \DataTypeTok{\{n:nat\}} 
\NormalTok{    (ch:chan(S,array(a)), data:arrref(a,n), len:int(n)): void = }\KeywordTok{let} 
    \KeywordTok{prval}\NormalTok{ () = unify ch }\CommentTok{(* locally interprets the quantifier *)}
      \KeywordTok{val}\NormalTok{ () = send (ch, len) }\CommentTok{(* provide an instance for the quantifier *)}
    \KeywordTok{fun}\NormalTok{ sendarr }\DataTypeTok{\{a:type\}} \DataTypeTok{\{n,m:nat|n<=m\}} 
\NormalTok{        (ch:chan(S,repeat(a,n)), x:int(n), data:arrref(a,m), len:int(m)): void = }
        \KeywordTok{if}\NormalTok{ x = }\DecValTok{0} \KeywordTok{then} \KeywordTok{let} \KeywordTok{prval}\NormalTok{ () = recurse ch}
                          \KeywordTok{prval}\NormalTok{ () = itef ch}
                       \KeywordTok{in}\NormalTok{ close ch }\KeywordTok{end} 
        \KeywordTok{else} \KeywordTok{let} \KeywordTok{prval}\NormalTok{ () = recurse ch}
                 \KeywordTok{prval}\NormalTok{ () = itet ch}
                   \KeywordTok{val}\NormalTok{ () = send (ch, data[len-x])}
              \KeywordTok{in}\NormalTok{ sendarr (ch, x-}\DecValTok{1}\NormalTok{, data, len) }\KeywordTok{end} 
\KeywordTok{in}\NormalTok{ sendarr (ch, len, data, len) }\KeywordTok{end}
\end{Highlighting}
\end{Shaded}
And its type is \[
\texttt{server} : \forall\tau{:}\sort{type}.\forall n{:}\sort{nat}.(\type{chan}(\srv,\stype{array}(\tau)), \type{arrref}(\tau,n), \type{int}(n))\rightarrow\type{1}
\] where \texttt{data} is the array to be sent, whose type is indexed by
the type of elements and the length of array. \texttt{len} is the length
of array, whose type is a singleton integer that equals the length of
\texttt{data}. \texttt{prval} denotes a proof value that has no runtime
semantics. After type-checking has passed, these values will be
eliminated.

\begin{example}[Queue] The example comes from {\tt SILL}\footnote{\url{https://github.com/ISANobody/sill}}, an implementation of binary session types based on \cite{Caires:2010gi}. As compared to a simple queue, we define a dependently typed queue indexed by its length as follows, with the higher-order $\fix$ introduced in \cref{sec:ext}, 
\end{example}
\[
\arraycolsep=1pt
\begin{array}{rl}
\stype{queue}(\tau{:}\sort{type}, n{:}\sort{int})\defeq
\branch(\cli, &\msg(\cli,\tau)\coloncolon\stype{queue}(\tau,n+1),\\ 
           &\ite(n>0, \msg(\srv,\tau)\coloncolon\stype{queue}(\tau,n-1),\nil(\srv)))
\end{array}
\] where the client can choose to either enqueue or dequeue an element
of type \(\tau\). In the dequeue case, instead of encoding an optional
value as a \texttt{branch} to deal with dequeuing from an empty queue,
we use the length of the queue to decide the continuation of the session
type. If the length \(n\) is greater than 0, the endpoint allows
dequeuing. Otherwise, the endpoint can only be closed. As mentioned
before, \texttt{itet}/\texttt{itef} are proof functions that have no
runtime cost, while a non-dependently session typed queue will require
\texttt{choose}/\texttt{offer} that need to communicate a tag at
runtime. We follow their example, and present the \texttt{elem} function
as follows, which given a queue and an element \texttt{e}, constructs a
new queue where \texttt{e} will be inserted into the queue as if it is
the first element, and \texttt{e} will be the first to be dequeued.

\begin{Shaded}
\begin{Highlighting}[]
\KeywordTok{fun}\NormalTok{ elem }\DataTypeTok{\{a:type\}} \DataTypeTok{\{n:nat\}} 
\NormalTok{    (q:chan(C,queue(a,n)), e:a): chan(C,queue(a,n+}\DecValTok{1}\NormalTok{)): void = }\KeywordTok{let}
        \CommentTok{(* out: endpoint held by the server}
\CommentTok{         * inp: endpoint to the tail of queue}
\CommentTok{         *)}
        \KeywordTok{fun}\NormalTok{ server }\DataTypeTok{\{n:nat\}}
\NormalTok{            (out:chan(S,queue(a,n+}\DecValTok{1}\NormalTok{)), inp:chan(C,queue(a,n))): void = }
            \KeywordTok{let} \KeywordTok{prval}\NormalTok{ () = recurse out }\CommentTok{(* unroll the fixpoint *)}
                  \KeywordTok{val}\NormalTok{  c = offer out   }
             \KeywordTok{in}\NormalTok{ case c }\KeywordTok{of}
                \CommentTok{(* dequeue case *)}
\NormalTok{                | Right () => }\KeywordTok{let} \KeywordTok{prval}\NormalTok{ () = itet  out}
                                    \KeywordTok{val}\NormalTok{ () = send (out, e)}
                               \CommentTok{(* let `inp` delegate the server *)}
                               \KeywordTok{in}\NormalTok{ cut (out, inp) }\KeywordTok{end} 
                \CommentTok{(* enqueue case *)}
\NormalTok{                |  Left () => }\KeywordTok{let}   \KeywordTok{val}\NormalTok{  y = recv    out}
                                  \KeywordTok{prval}\NormalTok{ () = recurse inp}
                                    \KeywordTok{val}\NormalTok{ () = choose (inp, Left())}
                                    \KeywordTok{val}\NormalTok{ () = send   (inp, y)}
                               \KeywordTok{in}\NormalTok{ server (out, inp) }\KeywordTok{end}
            \KeywordTok{end}
    \KeywordTok{in} 
        \CommentTok{(* create the server thread, and return the client endpoint *)}
\NormalTok{        create (}\KeywordTok{lam}\NormalTok{ out => server (out, queue)) }
    \KeywordTok{end}
\end{Highlighting}
\end{Shaded}

\begin{example}[Polymorphism]
\label[ex]{ex:poly}
We define a polymorphic cloud service that, given any unlimited function, will provide replicated services of such function. The example is taken from \cite{Caires:2013jb} that makes use of higher-order quantification over session types, and high-order sessions. We define polymorphic session types as follows,
\begin{align*}
\stype{service}(\pi{:}\sort{stype})&\defeq
\branch(\cli,\msg(\srv,\chan(\cli,\pi))\coloncolon\stype{service}(\pi), \nil(\cli))
\\
\stype{cloud}&\defeq
\quan(\cli,\lambda\pi{:}\sort{stype}.\msg(\cli,\chan(\srv,\pi)\rightarrow\type{1})\coloncolon\stype{service}(\pi))
\end{align*}
\end{example}

Here, \(\stype{service}(\pi)\) is a polymorphic session type constructor
that says a client can repeatedly choose to use a service through a
newly created endpoint disciplined by session type \(\pi\), or to close
it. \texttt{cloud} is a polymorphic session type that says, as long as
the client sends an \emph{unlimited/non-linear} function that can
provide the functionality described by \(\pi\), the server will turn it
into a replicated service. Corresponding server and client programs
could be written like the followings.

\begin{Shaded}
\begin{Highlighting}[]
\KeywordTok{implement}\NormalTok{ server (ch:chan(S,cloud)): void = }\KeywordTok{let} 
    \KeywordTok{prval}\NormalTok{ () = exify ch }\CommentTok{(* locally interpret `quan` as `exists` *)}
      \KeywordTok{val}\NormalTok{  f = recv ch  }\CommentTok{(* receive the witness and output it to the user *)}
    \CommentTok{(* the `srv` function provides replicated services }
\CommentTok{     * by spawning a new endpoint every time the user requests}
\CommentTok{     *)}
    \KeywordTok{fun}\NormalTok{ srv }\DataTypeTok{\{p:stype\}}\NormalTok{ (ch:chan(S,service(p)), f:chan(S,p)->void): void = }
        \KeywordTok{let} \KeywordTok{prval}\NormalTok{ () = recurse ch}
              \KeywordTok{val}\NormalTok{  c = offer ch}
        \KeywordTok{in}\NormalTok{ case c }\KeywordTok{of}
           \CommentTok{(* the user chooses to close *)}
\NormalTok{           | Right () => wait ch }
           \CommentTok{(* the user requests one such service *)}
\NormalTok{           |  Left () => }\KeywordTok{let} \KeywordTok{val}\NormalTok{ ep = create (}\KeywordTok{lam}\NormalTok{ ch => f ch)}
                             \KeywordTok{val}\NormalTok{ () = send (ch, ep)}
                          \KeywordTok{in}\NormalTok{ srv (ch, f) }\KeywordTok{end}
        \KeywordTok{end}
\KeywordTok{in} 
\NormalTok{    srv (ch, f)}
\KeywordTok{end}

\KeywordTok{implement}\NormalTok{ client (ch:chan(C,cloud)): void = }\KeywordTok{let} 
    \CommentTok{(* This is an instance of the service that does printing *)}
    \KeywordTok{fun}\NormalTok{ echo (ch:chan(S,msg(C,string)::}\KeywordTok{end}\NormalTok{(C))): void = }
        \KeywordTok{let} \KeywordTok{val}\NormalTok{ () = print (recv ch)}
         \KeywordTok{in}\NormalTok{ wait ch }\KeywordTok{end} 

    \KeywordTok{prval}\NormalTok{ () = unify ch }\CommentTok{(* locally interpret `quan` as `forall` *)}
      \KeywordTok{val}\NormalTok{ () = send (ch, echo) }\CommentTok{(* provide an instance *)}
    \CommentTok{(* request the printing service n times *)}
    \KeywordTok{fun}\NormalTok{ prt (ch:chan(C,service(msg(C,string)::}\KeywordTok{end}\NormalTok{(C))), n:int): void =}
        \KeywordTok{let} \KeywordTok{prval}\NormalTok{ () = recurse ch }
        \KeywordTok{in} \KeywordTok{if}\NormalTok{ n <= }\DecValTok{0} 
           \KeywordTok{then}\NormalTok{ (choose (ch, Right()); close ch) }
           \KeywordTok{else} \KeywordTok{let} \KeywordTok{val}\NormalTok{ () = choose (ch, Left())}
                    \CommentTok{(* receive the endpoint and use the service *)}
                    \KeywordTok{val}\NormalTok{ ep = recv ch }
                    \KeywordTok{val}\NormalTok{ () = send (ep, "hello world!")}
                    \KeywordTok{val}\NormalTok{ () = close ep }
                 \KeywordTok{in}\NormalTok{ prt (ch, n-}\DecValTok{1}\NormalTok{) }\KeywordTok{end} 
        \KeywordTok{end} 
\KeywordTok{in} 
\NormalTok{    prt (ch, }\DecValTok{10}\NormalTok{)}
\KeywordTok{end}  
\end{Highlighting}
\end{Shaded}

\vfill
\clearpage

\section{Appendix - Figures}\label{appendix---figures}

\input{rho} \input{mtlc0typing} \input{scx} \input{rho2}

\input{mtlcdeptyping2} \input{sessionapi} \input{mtlcdepevalctx}

\input{sessionreduction} \vfill
\rule[\baselineskip]{0pt}{\baselineskip}

\clearpage





\bibliography{./library}



\end{document}

%% file: mtlc0syntax.tex
\begin{figure}[h!]
\caption{Syntax of Multi-threaded $\lambda$-calculus with Linear Types}\label{fig:mtlc0syntax}
\[
\begin{array}{rrcl}
\text{types}                 \quad & \tau     & \defeq & \delta \mid \type{1} \mid \tau_1 \times \tau_2 \mid \hat\tau_1 \rightarrow \hat\tau_2                                   \\
\text{vtypes}                \quad & \hat\tau & \defeq & \hat\delta \mid \tau \mid \hat\tau_1\otimes\hat\tau_2 \mid \hat\tau_1\multimap\hat\tau_2                        \\
\text{dynamic constants}     \quad & \dcx     & \defeq & \dcc \mid \dcf                                                                                                  \\
\text{dynamic terms}         \quad & e        & \defeq & x \mid \dcx\,(\vv{e}) \mid \dcr \mid \langle \rangle \mid \langle e_1, e_2 \rangle \mid                         \\
                                   &          &        & \dyn{let}~\langle x_1,x_2 \rangle = e_1~\dyn{in}~e_2  \mid \dyn{if}~e~\dyn{then}~e_1~\dyn{else}~e_2 \mid        \\                                                          
                                   &          &        & \dyn{fst}(e) \mid \dyn{snd}(e) \mid \dyn{lam}~x.e \mid \dyn{app}(e_1, e_2)                                      \\
\text{dynamic values}        \quad & v        & \defeq & x \mid \dcc\,(\vv{v}) \mid \langle\rangle \mid \langle v_1,v_2 \rangle \mid \dyn{lam}~x.e                       \\
\text{dynamic type context}  \quad & \Gamma   & \defeq & \varnothing \mid \Gamma, x : \tau                                                                               \\
\text{dynamic vtype context} \quad & \Delta   & \defeq & \varnothing \mid \Delta, x : \hat\tau                                                                           \\
\text{dynamic signature}     \quad & \sig     & \defeq & \varnothing \mid \sig, \dcx: (\vv{\tau}) \Rightarrow \tau \mid \sig, \dcx: (\vv{\hat\tau}) \Rightarrow \hat\tau \\
\text{dynamic substitutions} \quad & \theta   & \defeq & [] \mid \theta[x \mapsto v]                                                                                     \\
\text{pools}                 \quad & \Pi      & \defeq & [] \mid \Pi[t \mapsto e]                                                                                       
\end{array}
\]
\end{figure}

%% file: statics.tex
\begin{figure}
\caption{Syntax of Statics}
\label{fig:statics}
\[
\begin{array}{rrcl}
\text{base sorts}           \quad & b      & \defeq & \sort{int} \mid \sort{bool} \mid \sort{type} \mid \sort{vtype} \\ 
\text{sorts}                \quad & \sigma & \defeq & b \mid \sigma_1 \rightarrow \sigma_2                                                \\
\text{static constants}     \quad & \scx   & \defeq & \scc \mid \scf                                                      \\
\text{static terms}         \quad & s      & \defeq & a \mid \scx\,(\vv{s}) \mid \lambda a {:} \sigma . s \mid s_1 (s_2)     \\
\text{static context}       \quad & \Sigma & \defeq & \varnothing \mid \Sigma, a : \sigma                                          \\
\text{static signature}     \quad & \sig   & \defeq & \varnothing \mid \sig, \scx : (\vv{\sigma}) \Rightarrow \sigma \\
\text{static substitutions} \quad & \theta & \defeq & [] \mid \theta[a \mapsto s]
\end{array}
\]
\end{figure}

%% file: types.tex
\begin{figure}
\caption{Types}
\label{fig:types}
\[
\begin{array}{rrcl}
\text{types}   \quad & \tau     & \defeq & a \mid \delta(\vv{s}) \mid \type{1} \mid \tau_1 \times \tau_2 \mid \hat\tau_1 \rightarrow \hat\tau_2 \mid P\guard\tau \mid P\assert\tau \mid \forall a {:} \sigma. \tau \mid \exists a {:} \sigma.\tau            \\
\text{vtypes}  \quad & \hat\tau & \defeq & a\mid\hat\delta(\vv{s}) \mid \tau \mid \hat\tau_1\otimes\hat\tau_2 \mid \hat\tau_1\multimap\hat\tau_2 \mid   P\guard\hat\tau \mid P\assert\hat\tau \mid \forall a{:}\sigma.\hat\tau \mid \exists a{:}\sigma.\hat\tau \\
\end{array}
\]
\end{figure}

%% file: mtlcdepsyntax.tex
\begin{figure}[!h]
\caption{Extended Dynamic Language Syntax}
\label{fig:mtlcdepsyntax}
\[
\begin{array}{rrcl}
\text{dynamic terms}  \quad & e & \defeq & \dotsb \mid \guardi(v) \mid \guarde(e) \mid  \assert(e) \mid \dyn{let}~{\assert}(x)=e_1~\dyn{in}~e_2 \mid\\
               						   &&& \foralli(v) \mid \foralle(e)  \mid \exists(e) \mid \dyn{let}~\exists(x)=e_1~\dyn{in}~e_2 \\
\text{dynamic values} \quad & v & \defeq & \dotsb \mid \guardi(v) \mid \foralli(v) \mid \assert(v) \mid \exists(v)   \\
\end{array}
\]
\end{figure}

%% file: mtlcdeptyping.tex
\begin{figure}[!h]
\caption{Some Additional Typing Rules of $\Ldep$}
\label{fig:mtlcdeptyping}

\begin{gather*}
\prftree[r]{\bf ty-$\forall$-intr}
	{\Sigma,a:\sigma;\vv P;\Gamma;\Delta\vdash v:\hat\tau}
	{\Sigma;\vv P;\Gamma;\Delta\vdash\foralli(v):\forall a{:}\sigma.\hat\tau}
\qquad
\prftree[r]{\bf ty-$\forall$-elim}
	{\stackbox{
		$\Sigma\vdash s:\sigma$\\
		$\Sigma;\vv P;\Gamma;\Delta\vdash e:\forall a{:}\sigma.\hat\tau$
	}}
	{\Sigma;\vv P;\Gamma;\Delta\vdash\foralle(e):\hat\tau[a\mapsto s]}
\\
\prftree[r]{\bf ty-$\exists$-intr}
	{\stackbox{
		$\Sigma\vdash s:\sigma$\\
		$\Sigma;\vv P;\Gamma;\Delta\vdash e:\hat\tau[a\mapsto s]$
	}}
	{\Sigma;\vv P;\Gamma;\Delta\vdash\exists(e):\exists a{:}\sigma.\hat\tau}
\quad
\prftree[r]{\bf ty-$\exists$-elim}
	{\stackbox{
		$\Sigma;\vv P;\Gamma;\Delta\vdash e_1:\exists a{:}\sigma.\hat\tau_1$\\
		$\Sigma,a:\sigma;\vv P;(\Gamma;\Delta),x:\hat\tau_1\vdash e_2:\hat\tau_2$
	}}
	{\Sigma;\vv P;\Gamma;\Delta\vdash\dyn{let}~\exists(x)=e_1~\dyn{in}~e_2 : \hat\tau_2}
\end{gather*}
\end{figure}

%% file: sessionsyntax.tex
\begin{figure}[!h]
\caption{Syntax of Dependent Session Types}
\label{fig:sessionsyntax}
\[
\begin{array}{rrcl}
\text{base sorts} & b   & \defeq & \dotsb \mid \sort{stype} \\
\text{stypes}     & \pi & \defeq & \stype{end}(i) \mid \stype{msg}(i, \hat\tau) \coloncolon \pi \mid \stype{branch}(i, \pi_1, \pi_2) \mid \\
                  &     &        & \stype{ite}(b,\pi_1,\pi_2) \mid \stype{quan}(i,\lambda a{:}\sigma.\pi) \mid \stype{fix}(\lambda a{:}\sort{stype}.\pi) 
\end{array}
\]
\end{figure} 

%% file: elixirdelegate.tex
\begin{figure}[h!]
\caption{Example {\tt cut} in Erlang/Elixir Back-end}
\label{fig:elixir}
\centering
\begin{tikzpicture}[node distance=1.2cm]

\tikzstyle{owner}=[rectangle]
\tikzstyle{cut}=[circle,dotted]
\tikzstyle{endpoint}=[circle]
\tikzstyle{every node}=[thick,inner sep=0pt,draw=black!75,minimum size=6mm]

\begin{scope}
    \node [owner] (p2) {$P_2$};
    \node [owner] (p1) [left of=p2,xshift=-0.5cm]{$P_1$};
    \node [owner] (p3) [right of=p2,xshift=0.5cm] {$P_3$};

    \node [endpoint] (a0) [below of=p1] {$A$};
    \node [cut] (a1) [below of=p2,xshift=-0.5cm] {$A'$};
    \node [cut] (b0) [below of=p2,xshift=0.5cm] {$B'$};
    \node [endpoint] (b1) [below of=p3] {$B$};

    \draw [latex'-latex'] (p1) -- (a0);
    \draw [latex'-latex'] (p2) -- (a1);
    \draw [latex'-latex'] (p2) -- (b0);
    \draw [latex'-latex'] (p3) -- (b1);
    \draw [latex'-latex',double] (a0) -- (a1);
    \draw [latex'-latex',double] (b0) -- (b1);
\end{scope}
    
\begin{scope}[xshift=5.5cm]
    \node [owner] (p1) {$P_1$};
    \node [owner] (p3) [right of=p1] {$P_3$};
    \node [owner] (p2) [right of=p3] {$P_2$};

    \node [endpoint] (a0) [below of=p1] {$A$};
    \node [endpoint] (b0) [below of=a0] {$B'$};
    \node [endpoint] (b1) [below of=p3] {$B$};
    \node [endpoint] (a1) [below of=b1] {$A'$};

    \draw [latex'-latex'] (p1) -- (a0);
    \draw [-latex',dashed] (a0) -- (b0);
    \draw [latex'-latex'] (p3) -- (b1);
    \draw [-latex',dashed] (b1) -- (a1);
    \draw [latex'-latex',double] (a0) -- (a1);
    \draw [latex'-latex',double] (b0) -- (b1);

\end{scope}

\draw [->] (2.8,-1.2) -- (4.5,-1.2) node [draw=none,midway,above] {\tt{cut}$(A',B')$};
 
\end{tikzpicture}
\end{figure}
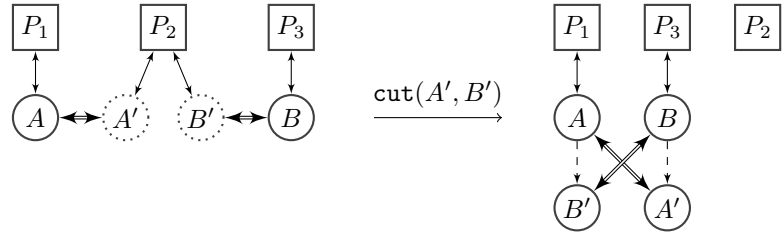

%% file: rho.tex
\begin{figure}[h!]
\caption{Definition of $\rho(\cdot)$ in $\Lzero$}
\label{fig:rho}
\[
\begin{array}{rclrcl}
\rho(\dyn{fst}(e))                                      &=& \rho(e)                              &
\rho(x)                                                 &=& \varnothing                          \\
\rho(\dyn{snd}(e))                                      &=& \rho(e)                              &
\rho(\dcr)                                              &=& \{\textit{dcr}\}                     \\
\rho(\dyn{lam}~x.e)                                     &=& \rho(e)                              &
\rho(\dcx\,(e_1,\dotsc,e_n))                            &=& \rho(e_1)\uplus\dotsb\uplus\rho(e_n) \\
\rho(\dyn{app}(e_1,e_2))                                &=& \rho(e_1)\uplus\rho(e_2)             &
\rho(\langle e_1,e_2\rangle)                            &=& \rho(e_1)\uplus\rho(e_2)             \\
\rho(\dyn{let}~\langle x_1,x_2\rangle=e_1~\dyn{in}~e_2) &=& \rho(e_1)\uplus\rho(e_2)             &
\rho(\langle\rangle)                                    &=& \varnothing                          \\
\rho(\dyn{if}~e~\dyn{then}~e_1~\dyn{else}~e_2)          &=& \rho(e)\uplus\rho(e_1)             &  \\
\rho(\Pi)    &=& \multicolumn{4}{l}{\biguplus_{t}\rho(\Pi(t))~t\in\dom(\Pi)}                  \\
\rho(\theta) &=& \multicolumn{4}{l}{\biguplus_{x}\rho(\theta(x))~x\in\dom(\theta)}
\end{array}
\]
\end{figure}

%% file: mtlc0typing.tex
\begin{figure}[h!]
\caption{Typing Rules of $\Lzero$}
\label{fig:mtlc0typing}
\addtolength{\jot}{.1em}
\begin{gather*}
\prftree[r]{\bf ty-res}
	{\sig\vDash\dcr:\hat\delta}
	{\Gamma;\varnothing\vdash\dcr:\hat\delta}
\quad
\prftree[r]{\bf ty-cst}
	{\stackbox{
		$\sig\vDash\dcx:(\hat\tau_1,\dotsc,\hat\tau_n)\Rightarrow\hat\tau$\\
		$\Gamma;\Delta_i\vdash e_i:\hat\tau_i\quad1\leqslant i\leqslant n$}}
	{\Gamma;\Delta_1,\dotsc,\Delta_n\vdash\dcx\,(e_1,\dotsc,e_n):\hat\tau}
\\
\prftree[r]{\bf ty-var-i}
	{\Gamma, x:\tau;\varnothing\vdash x:\tau}
\quad
\prftree[r]{\bf ty-var-l}
	{\Gamma;\Delta,x:\hat\tau\vdash x:\hat\tau}
\quad
\prftree[r]{\bf ty-unit}
	{\Gamma;\varnothing\vdash\langle\rangle:\type{1}}
\\
\prftree[r]{\bf ty-tup-i}
	{\Gamma;\Delta_1\vdash e_1:\tau_1}
	{\Gamma;\Delta_2\vdash e_2:\tau_2}
	{\Gamma;\Delta_1,\Delta_2\vdash \langle e_1,e_2\rangle:\tau_1\times\tau_2}
\quad
\prftree[r]{\bf ty-tup-l}
	{\Gamma;\Delta_1\vdash e_1:\hat\tau_1}
	{\Gamma;\Delta_2\vdash e_2:\hat\tau_2}
	{\Gamma;\Delta_1,\Delta_2\vdash\langle e_1,e_2\rangle:\hat\tau_1\otimes\hat\tau_2}
\\
\prftree[r]{\bf ty-fst}
	{\Gamma;\Delta\vdash e:\tau_1\times\tau_2}
	{\Gamma;\Delta\vdash \dyn{fst}(e):\tau_1}
\quad
\prftree[r]{\bf ty-snd}
	{\Gamma;\Delta\vdash e:\tau_1\times\tau_2}
	{\Gamma;\Delta\vdash \dyn{snd}(e):\tau_2}
\\
\prftree[r]{\bf ty-tup-elim}
	{\Gamma;\Delta_1\vdash e_1:\hat\tau_1\otimes\hat\tau_2}
	{\Gamma;\Delta_2,x_1:\hat\tau_1,x_2:\hat\tau_2\vdash e_2:\hat\tau}
	{\Gamma;\Delta_1,\Delta_2\vdash\dyn{let}~\langle x_1,x_2\rangle=e_1~\dyn{in}~e_2:\hat\tau}
\\
\prftree[r]{\bf ty-lam-i}
	{(\Gamma;\varnothing),x:\hat\tau_1\vdash e:\hat\tau_2}
	{\rho(e)=\varnothing}
	{\Gamma;\varnothing\vdash\dyn{lam}~x.e:\hat\tau_1\rightarrow\hat\tau_2}
\quad
\prftree[r]{\bf ty-lam-l}
	{(\Gamma;\Delta),x:\hat\tau_1\vdash e:\hat\tau_2}
	{\Gamma;\Delta\vdash\dyn{lam}~x.e:\hat\tau_1\multimap\hat\tau_2}
\\
\prftree[r]{\bf ty-app-i}
	{\stackbox{
		$\Gamma;\Delta_2\vdash e_2:\hat\tau_1$\\
		$\Gamma;\Delta_1\vdash e_1:\hat\tau_1\rightarrow\hat\tau_2$}}
	{\Gamma;\Delta_1,\Delta_2\vdash\dyn{app}(e_1,e_2):\hat\tau_2}
\quad
\prftree[r]{\bf ty-app-l}
	{\stackbox{
		$\Gamma;\Delta_2\vdash e_2:\hat\tau_1$\\
		$\Gamma;\Delta_1\vdash e_1:\hat\tau_1\multimap\hat\tau_2$ }}
	{\Gamma;\Delta_1,\Delta_2\vdash\dyn{app}(e_1,e_2):\hat\tau_2}
\\
\prftree[r]{\bf ty-if}
	{\Gamma;\Delta_1\vdash e:\type{bool}}
	{\Gamma;\Delta_2\vdash e_1:\hat\tau}
	{\Gamma;\Delta_2\vdash e_2:\hat\tau}
	{\rho(e_1)=\rho(e_2)}
	{\Gamma;\Delta_1,\Delta_2\vdash\dyn{if}~e~\dyn{then}~e_1~\dyn{else}~e_2:\hat\tau}
\\
\prftree[r]{\bf ty-pool}
	{\varnothing;\varnothing\vdash\Pi(0):\hat\tau}
	{\varnothing;\varnothing\vdash\Pi(t):\type{1}~\textrm{for each}~t\in\dom(\Pi)\backslash\{0\}}
	{\varnothing;\varnothing\vdash\Pi:\hat\tau}
\end{gather*}
\end{figure}

%% file: scx.tex
\begin{figure}
\caption{Some Static Constants ($\scc$) in $\Ldep$}
\label{fig:scx}
\[
\begin{array}{rclrcl}
\times      & : & (\sort{type}, \sort{type})   \Rightarrow \sort{type}     & \otimes     & : & (\sort{vtype}, \sort{vtype}) \Rightarrow \sort{vtype}    \\
\rightarrow & : & (\sort{vtype}, \sort{vtype}) \Rightarrow \sort{type}       & \multimap   & : & (\sort{vtype}, \sort{vtype}) \Rightarrow \sort{vtype} \\
\guard      & : & (\sort{bool}, \sort{type}) \Rightarrow \sort{type}       & \guard      & : & (\sort{bool}, \sort{vtype}) \Rightarrow \sort{vtype} \\
\assert     & : & (\sort{bool}, \sort{type}) \Rightarrow \sort{type}       & \assert     & : & (\sort{bool}, \sort{vtype}) \Rightarrow \sort{vtype} \\
\forall     & : & (\sigma \rightarrow \sort{type}) \Rightarrow \sort{type} & \forall     & : & (\sigma \rightarrow \sort{vtype}) \Rightarrow \sort{vtype} \\
\exists     & : & (\sigma \rightarrow \sort{type}) \Rightarrow \sort{type} & \exists     & : & (\sigma \rightarrow \sort{vtype}) \Rightarrow \sort{vtype} \\
\type{int}  & : & () \Rightarrow \sort{type}                               & \type{int}  & : & (\sort{int}) \Rightarrow \sort{type} \\
\type{bool} & : & () \Rightarrow \sort{type}                               & \type{bool} & : & (\sort{bool}) \Rightarrow \sort{type} \\
\top        & : & () \Rightarrow \sort{bool}                               & \bot        & : & () \Rightarrow \sort{bool} \\
\subtype    & : & (\sort{type}, \sort{type}) \Rightarrow \sort{bool}       & \subtype    & : & (\sort{vtype}, \sort{vtype}) \Rightarrow \sort{bool} \\
\type{1}    & : & () \Rightarrow \sort{type}
\end{array}
\]
\end{figure}

%% file: rho2.tex
\begin{figure}[h!]
\caption{Additional Definition of $\rho(\cdot)$ in $\Ldep$}
\label{fig:rho2}
\[
\begin{array}{rclrcl}
\rho(\guardi(v))                                      &=& \rho(v)                              &
\rho(\foralli(v))                                                 &=& \rho(v)                          \\
\rho(\guarde(e))                                      &=& \rho(e)                              &
\rho(\foralle(e))                                              &=& \rho(e)                     \\
\rho(\assert(e))                                     &=& \rho(e)                              &
\rho(\exists(e))                            &=& \rho(e) \\
\rho(\dyn{let}~\assert(x)=e_1~\dyn{in}~e_2) &=& \rho(e_1)\uplus\rho(e_2)             &
\rho(\dyn{let}~\exists(x)=e_1~\dyn{in}~e_2) &=& \rho(e_1)\uplus\rho(e_2)             
\end{array}
\]
\end{figure}

%% file: mtlcdeptyping2.tex
\begin{figure}[!h]
\caption{Additional Typing Rules of $\Ldep$}
\label{fig:mtlcdeptyping2}
\begin{gather*}
\prftree[r]{\bf ty-$\forall$-intr}
	{\Sigma,a:\sigma;\vv P;\Gamma;\Delta\vdash v:\hat\tau}
	{\Sigma;\vv P;\Gamma;\Delta\vdash\foralli(v):\forall a{:}\sigma.\hat\tau}
\qquad
\prftree[r]{\bf ty-$\forall$-elim}
	{\stackbox{
		$\Sigma\vdash s:\sigma$\\
		$\Sigma;\vv P;\Gamma;\Delta\vdash e:\forall a{:}\sigma.\hat\tau$
	}}
	{\Sigma;\vv P;\Gamma;\Delta\vdash\foralle(e):\hat\tau[a\mapsto s]}
\\
\prftree[r]{\bf ty-$\exists$-intr}
	{\stackbox{
		$\Sigma\vdash s:\sigma$\\
		$\Sigma;\vv P;\Gamma;\Delta\vdash e:\hat\tau[a\mapsto s]$
	}}
	{\Sigma;\vv P;\Gamma;\Delta\vdash\exists(e):\exists a{:}\sigma.\hat\tau}
\quad
\prftree[r]{\bf ty-$\exists$-elim}
	{\stackbox{
		$\Sigma;\vv P;\Gamma;\Delta\vdash e_1:\exists a{:}\sigma.\hat\tau_1$\\
		$\Sigma,a:\sigma;\vv P;(\Gamma;\Delta),x:\hat\tau_1\vdash e_2:\hat\tau_2$
	}}
	{\Sigma;\vv P;\Gamma;\Delta\vdash\dyn{let}~\exists(x)=e_1~\dyn{in}~e_2 : \hat\tau_2}
\\
\prftree[r]{\bf ty-$\guard$-intr}
	{\Sigma;\vv P,P';\Gamma;\Delta\vdash e:\hat\tau}
	{\Sigma;\vv P;\Gamma;\Delta\vdash \guardi(e):P'\guard\hat\tau}
\quad
\prftree[r]{\bf ty-$\guard$-elim}
	{\Sigma;\vv P;\Gamma;\Delta\vdash e:P'\guard\hat\tau}
	{\Sigma;\vv P\vdash P'}
	{\Sigma;\vv P;\Gamma;\Delta\vdash\guarde(e):\hat\tau}
\\
\prftree[r]{\bf ty-$\assert$-intr}
	{\stackbox{
		$\Sigma;\vv P\vdash P'$\\
		$\Sigma;\vv P;\Gamma;\Delta\vdash e:\hat\tau$
	}}
	{\Sigma;\vv P;\Gamma;\Delta\vdash \assert(e) : P'\assert\hat\tau}
\quad
\prftree[r]{\bf ty-$\assert$-elim}
	{\stackbox{
		$\Sigma;\vv P;\Gamma;\Delta\vdash e_1:P'\assert\hat\tau_1$\\
		$\Sigma;\vv P,P';(\Gamma;\Delta),x:\hat\tau\vdash e_2:\hat\tau_2$
	}}
	{\Sigma;\vv P;\Gamma;\Delta\vdash\dyn{let}~\assert(x)=e_1~\dyn{in}~e_2:\hat\tau_2}
\\
\prftree[r]{\bf ty-sub}
	{\Sigma;\vv P;\Gamma;\Delta\vdash e:\hat\tau_1}
	{\Sigma;\vv P\vdash\hat\tau_1\subtype\hat\tau_2}
	{\Sigma;\vv P;\Gamma;\Delta\vdash e:\hat\tau_2}
\end{gather*}
\end{figure}

%% file: sessionapi.tex
\begin{figure}[h!]
\caption{Extended Dynamic Constants in $\Lchan$}
\label{fig:sessionapi}
\begin{align*}
\texttt{create}  &: \forall r_1,r_2{:}\sort{role}.\forall\pi{:}\sort{stype}.(r_1\neq r_2)\guard(\type{chan}(r_2,\pi)\multimap\type{1})\Rightarrow\type{chan}(r_1,\pi) \\
\texttt{send}    &: \forall r,r_0{:}\sort{role}.\forall\pi{:}\sort{stype}.\forall\hat\tau{:}\sort{vtype}.
					\\&\qquad\qquad (r=r_0)\guard(\chan(r,\msg(r_0,\hat\tau)\coloncolon\pi), \hat\tau)\Rightarrow\chan(r,\pi) \\
\texttt{recv}    &: \forall r,r_0{:}\sort{role}.\forall\pi{:}\sort{stype}.\forall\hat\tau{:}\sort{vtype}.
					\\&\qquad\qquad (r\neq r_0)\guard\chan(r,\msg(r_0,\hat\tau\coloncolon\pi))\Rightarrow\hat\tau\otimes\chan(r,\pi) \\
\texttt{close}   &: \forall r,r_0{:}\sort{role}.(r=r_0)\guard\chan(r,\stype{end}(r_0))\Rightarrow\type{1} \\
\texttt{wait}    &: \forall r,r_0{:}\sort{role}.(r\neq r_0)\guard\chan(r,\stype{end}(r_0))\Rightarrow\type{1} \\
\texttt{offer}   &: \forall r,r_0{:}\sort{role}.\forall \pi_1,\pi_2{:}\sort{stype}.(r\neq r_0)\guard\chan(r,\stype{branch}(r_0,\pi_1,\pi_2)
                              \\&\qquad\qquad \Rightarrow\exists b{:}\sort{bool}.\type{bool}(b)\otimes\chan(r,\ite(b,\pi_1,\pi_2))\\
\texttt{choose}  &: \forall r,r_0{:}\sort{role}.\forall \pi_1,\pi_2{:}\sort{stype}.\forall b{:}\sort{bool}.(r=r_0) \guard(\chan(r,\stype{branch}(r_0,\pi_1,\pi_2)), \type{bool}(b)) 
                              \\&\qquad\qquad \Rightarrow \chan(r,\ite(b,\pi_1,\pi_2)) \\
\texttt{unify}   &: \forall r,r_0{:}\sort{role}.\forall\pi{:}\sort{stype}.\forall f{:}\sigma\rightarrow\sort{stype}.
                              \\&\qquad\qquad (r= r_0)\guard\chan(r,\quan(r_0,f))\Rightarrow\forall s{:}\sigma.\chan(r,f(s)) \\
\texttt{exify}   &: \forall r,r_0{:}\sort{role}.\forall\pi{:}\sort{stype}.\forall f{:}\sigma\rightarrow\sort{stype}.
                              \\&\qquad\qquad (r\neq r_0)\guard\chan(r,\quan(r_0,f))\Rightarrow\exists s{:}\sigma.\chan(r,f(s)) \\
\texttt{itet}    &: \forall r{:}\sort{role}.\forall \pi_1,\pi_2{:}\sort{stype}.\chan(r,\stype{ite}(\top,\pi_1,\pi_2))\Rightarrow\chan(r,\pi_1)\\
\texttt{itef}    &: \forall r{:}\sort{role}.\forall \pi_1,\pi_2{:}\sort{stype}.\chan(r,\stype{ite}(\bot,\pi_1,\pi_2))\Rightarrow\chan(r,\pi_2)\\
\texttt{recurse} &: \forall r{:}\sort{role}.\forall f{:}\sort{stype}\rightarrow\sort{stype}.\chan(r,\fix(f))\Rightarrow\chan(r,f(\fix(f))) \\
\texttt{cut}     &: \forall r_1,r_2{:}\sort{role}.\forall\pi{:}\sort{stype}.(r_1\neq r_2)\guard(\chan(r_1,\pi),\chan(r_2,\pi))\Rightarrow\type{1}
\end{align*}
\end{figure}

%% file: mtlcdepevalctx.tex
\begin{figure}[h!]
\caption{Additional Evaluation Context for $\Ldep$}
\label{fig:mtlcdepevalctx}
\begin{align*}
\text{evaluation context}~E \defeq~& \dotsb \mid \guarde(E) \mid \foralle(E) \mid \\
								   &\assert(E) \mid \dyn{let}~\assert(x)=E~\dyn{in}~e \mid\\
                                   & \exists(E) \mid \dyn{let}~\exists(x)=E~\dyn{in}~e 
\end{align*}
\end{figure}

%% file: sessionreduction.tex
\begin{figure}[h!]
\caption{Reductions on Pools in $\Lchan$, Part A}
\label{fig:sessionreduction}
\addtolength{\jot}{.1em}
To distinguish linear channels, we assign a natural number $i$ to each channel as an identifier. We use $ch$ to range over linear channels, $ch_i$ for a channel with identifier $i$, and $ch_{i,r_1}$/$ch_{i,r_2}$ for its dual endpoints of role $r_1$/$r_2$, respectively. Assuming $i$ is some channel identifier and $r_1,r_2$ are two different roles. Assuming $v$ is some value, $b$ is some boolean value.
\begin{gather*}
\prftree[r]{\bf pr-create}
  {\Pi(t)=E[\texttt{create}(\dyn{lam}~x.e)]}
  {\Pi\rightarrow\Pi[t\colonequals E[ch_{i,r_2}]][t'\mapsto\dyn{app}(\dyn{lam}~x.e, ch_{i,r_1})]}
\\
\prftree[r]{\bf pr-end}
  {\Pi(t_1)=E[\texttt{close}(ch_{i,r_1})]}
  {\Pi(t_2)=E[\texttt{wait}(ch_{i,r_2})]}
  {\Pi\rightarrow\Pi[t_1\colonequals E[\langle\rangle]][t_2\colonequals E[\langle\rangle]]}
\\
\prftree[r]{\bf pr-msg}
  {\Pi(t_1)=E[\texttt{send}(ch_{i,r_1}, v)]}
  {\Pi(t_2)=E[\texttt{recv}(ch_{i,r_2})]}
  {\Pi\rightarrow\Pi[t_1\colonequals E[ch_{i,r_1}]][t_2\colonequals E[\langle v,ch_{i,r_2}\rangle]]}
\\
\prftree[r]{\bf pr-branch}
  {\Pi(t_1)=E[\texttt{choose}(ch_{i,r_1}, b)]}
  {\Pi(t_2)=E[\texttt{offer}(ch_{i,r_2})]}
  {\Pi\rightarrow\Pi[t_1\colonequals E[ch_{i,r_1}]][t_2\colonequals E[\langle b,ch_{i,r_2}\rangle]]}
\end{gather*}
\end{figure}

\begin{figure}[h!]
\caption{Reductions on Pools in $\Lchan$, Part B, \tt cut}
\label{fig:sessionreductioncut}
\addtolength{\jot}{.1em}
\begin{gather*}
\text{Let $e$ be $\texttt{cut}(ch_{i,r_2}, ch_{j,r_1})$, $r_1\neq r_2$, and $i\neq j$}\\
\\
\prftree[r]{\bf pr-cut-end}
  {\Pi(t_1)=E[\texttt{close}(ch_{i,r_1})]}
  {\Pi(t)=E[e]}
  {\Pi(t_2)=E[\texttt{wait}(ch_{j,r_2})]}
  {\Pi\rightarrow\Pi[t_1\colonequals E[\langle\rangle]][t\colonequals E[\langle\rangle]][t_2\colonequals E[\langle\rangle]]}
\\
\prftree[r]{\bf pr-cut-msg}
  {\Pi(t_1)=E[\texttt{send}(ch_{i,r_1}, v)]}
  {\Pi(t)=E[e]}
  {\Pi(t_2)=E[\texttt{recv}(ch_{j,r_2})]}
  {\Pi\rightarrow\Pi[t_1\colonequals E[ch_{i,r_1}]][t\colonequals E[e]][t_2\colonequals E[\langle v,ch_{j,r_2}\rangle]]}
\\
\prftree[r]{\bf pr-cut-branch}
  {\Pi(t_1)=E[\texttt{choose}(ch_{i,r_1}, b)]}
  {\Pi(t)=E[e]}
  {\Pi(t_2)=E[\texttt{offer}(ch_{j,r_2})]}
  {\Pi\rightarrow\Pi[t_1\colonequals E[ch_{i,r_1}]][t\colonequals E[e]][t_2\colonequals E[\langle b,ch_{j,r_2}\rangle]]}
\end{gather*}
\end{figure}